\documentclass{article}
\usepackage[T1]{fontenc}
\usepackage{fullpage}
\usepackage{todonotes}
\usepackage{graphicx}
\usepackage{amsmath,xcolor,comment,booktabs,amssymb,amsthm}
\usepackage{enumitem}
\usepackage{hyperref}
\usepackage[table]{xcolor}
\usepackage{soul}
\usepackage{natbib}
\usepackage{nicefrac}
\usepackage[table]{xcolor}

\newtheorem{assumption}{Assumption}
\newtheorem{observation}{Observation}
\newtheorem{definition}{Definition}
\newtheorem{proposition}{Proposition}
\newtheorem{theorem}{Theorem}
\newtheorem{example}{Example}
\newtheorem{lemma}{Lemma}
\newtheorem{remark}{Remark}
\DeclareMathOperator{\E}{\mathbb{E}}
\newcommand{\TM}[2]{#1^{(#2)}}
\newcommand{\tm}[1]{\TM{#1}{t}}

\begin{document}
	\title{Blockchain Monetary Policy, Strategic Participation, and the Dynamics of Decentralization}
	

	%%%% non-anonymous submission below %%%%%
	\author{Diodato Ferraioli\thanks{University of Salerno, Email: {\tt dferraioli@unisa.it}.} \and Paolo Penna\thanks{IOG, Email: {\tt paolo.penna@iohk.io}.} \and Manvir Schneider\thanks{Cardano Foundation, Email: Email: {\tt manvir.schneider@cardanofoundation.org}.} \and Carmine Ventre\thanks{King's College London, Email: Email: {\tt carmine.ventre@kcl.ac.uk}.}}
	
	\date{}
	%%%%%% 
	
	%%%% anonymous submission %%%%
	%\author{}
	%\institute{}
	
	\maketitle              % typeset the header of the contribution

\begin{abstract}
	Blockchain systems must navigate a fundamental tradeoff between rewarding capable participants for performance and preserving a sufficiently decentralized distribution of stake. We study this tradeoff through algorithmic reward policies in dynamic blockchain environments with publicly known capability types and evolving token stakes. 

We analyze two settings. First, under full participation, we study virtual-stake policies that combine type and stake in winner selection. We show that fixed virtual-stake rules exhibit strong path dependence, and that this limitation extends to broader winner-selection rules based on virtual stake; when the type weight varies over time, the winner distribution asymptotically approaches either stake-proportional or type-proportional selection.

Second, we study a repeated participation game in which agents may abstain, rewards update stake, and token value depends on the decentralization of the participating stake. We characterize stage equilibria under myopic behavior.
% and under asymmetric lookahead. 
In this environment, strongly performance-favoring winner-take-most policies can drive the system to minimum decentralization under myopic play. 

We then identify policies that preserve participation or decentralization only by departing from strict performance-first reward design. In particular, our analysis identifies three benchmarks: policies that spread winner eligibility across the top-$k$ types, policies that simulate lookahead, and broad-sharing reward rules. Each mitigates failures of performance-first designs with different tradeoffs in participation, robustness, and performance priority. Overall, our results suggest that winner selection alone does not robustly resolve the capability–stake tradeoff: sustaining decentralization requires reward rules that react explicitly to concentration, even when this entails sacrificing some immediate performance advantage.
\end{abstract}

%%% Use this command to specify a few keywords describing your work.
%%% Keywords should be separated by commas.

%%%%%%%%%%%%%%%%%%%%%%%%%%%%%%%%%%%%%%%%%%%%%%%%%%%%%%%%%%%%%%%%%%%%%%%%

%%% Include any author-defined commands here.

%%% The following commands remove the headers in your paper. For final 
%%% papers, these will be inserted during the pagination process.

%\pagestyle{fancy}
%\fancyhead{}

%%% The next command prints the information defined in the preamble.

\maketitle

\section{Introduction}
Blockchain reward design must balance two 
often competing
forces:
% that are often in tension: 
rewarding more capable participants to improve performance, and preventing stake from concentrating so
% strongly 
that decentralization erodes. This tradeoff appears across validator and block-production settings, where protocols
% rules 
affect not only who performs tasks 
% efficiently
today but also how stake and influence are distributed tomorrow. Tokenomics policies must navigate this tradeoff 
% explicitly
by deciding who to select for the task
% at hand
(the winner) and how to reward the participating agents.

We study this tradeoff in a stylized dynamic model with two state variables for each player: a publicly known capability type and an evolving token stake. Our analysis separates two questions. First, if participation is fixed, can winner-selection rules that combine type and stake correct early imbalances over time? Second, when participation is endogenous and token value responds to decentralization of the participating stake, which reward rules can sustain both performance and participation incentives?

The permissionless nature of blockchain gives rise to participation games %, as first observed in~
\cite{chaidos2023blockchain}. Players can decide to participate (e.g., run the software) or abstain, and selection and reward rules must be engineered so that the system remains healthy. Unlike~\cite{chaidos2023blockchain}, which focuses on single-shot %monetary 
rewards, our work explores the multi-dimensional nature of system objectives in a multi-round game where participants are awarded tokens whose value can change over time. %, for instance if the system becomes overly centralized.

Our goal is to formally address this tension and analyze how different monetary policies impact system performance, decentralization, and token value over time. Beyond being of theoretical interest, this matters for protocol engineering and market design. For example, in proof-of-stake systems, selection and reward rules shape validator participation and 
% how stake concentrates over time
stake concentration; this in turn affects decentralization and the \emph{economic security} of the system by changing the effective stake needed for coordinated attacks on the participating set. At the same time, decentralization is also a key input into \emph{legitimacy and valuation}: perceived concentration can affect users' and investors' confidence and thus token value, which feeds back into participants' incentives. Our results can therefore be read both as design constraints on performance-first incentives and as guidance on when 
% a protocol must 
to
explicitly react to concentration rather than relying on passive stake-proportional dynamics.

Importantly, we do not aim to 
% provide a literal
model
% of 
any single deployed consensus protocol, but to isolate a core economic tension between capability-sensitive selection and stake-sensitive decentralization. 
%To keep the model tractable, we assume public types, binary participation decisions, a fixed reward budget, and a token-value channel driven by decentralization of the participating stake. 
% These assumptions make it possible 
%This allows
%to 
In a tractable model, we obtain sharp characterizations of both dynamic allocation and strategic participation,
% while also 
clarifying the limits of performance-first reward design.

%\paragraph*
 \smallskip 
\noindent \emph{Our contributions.}
% We study how blockchain reward policies navigate the tension between capability and stake. 
% Our framework 
We
track
two state variables for each player -- a publicly known capability type and an evolving token stake -- and 
study how blockchain reward policies navigate the tension between them, by
asking
how winner selection and reward allocation shape long-run decentralization. 
We make three contributions.
\begin{description}[topsep=3pt,itemsep=1pt,leftmargin=0pt]
	\item[Full-participation benchmark.] We first isolate the effect of winner selection under full participation. For fixed virtual stake, winner probabilities become invariant after the first round (Proposition~\ref{prop:pinvariant}), implying strong path dependence: early imbalances persist, and for moderate type weights the highest-type player may obtain only a negligible long-run share of stake (Theorem~\ref{thm:virtual}). This limitation extends to the broader family of policies that weigh virtual stakes more heavily; there are initial stake profiles under which the highest-type player obtains an asymptotically negligible share with arbitrarily high probability (Theorem \ref{thm:gamma-lockin}).
    % Finally, 
    When the type weight varies over time, the policy asymptotically collapses either to stake-proportional 
    or to 
    type-proportional selection (Proposition \ref{prop:degeneracy-lambda} and %Observation \ref{obs:stake-martingale}, 
    Theorem \ref{thm:type-proportional-limit}).
    \item[Endogenous participation.] We then study a repeated participation game in which token value depends on the decentralization of the participating stake. Under myopic behavior, 
    % and asymmetric lookahead,
    we characterize stage equilibria and show that they take a suffix form ordered by stake (Lemmas~\ref{lem:equilibria}). %and \ref{lem:equilibria:lookahead}). 
    This gives a tractable equilibrium structure for analyzing strategic participation.
	
	\item[Limits of performance-first design and corrective policies.]  For a general class of type-fa\-vo\-ring policies, the system eventually reaches a low-decentralization region and repeatedly returns to it under the $\tau$-decentralization index (Theorem \ref{thm:fail_generic}).
    % In the extreme case where 
    When
    the winner selection is nearly-deterministic, performance-first design can fail completely, driving the system to minimum decentralization under myopic play (Theorem \ref{thm:fails_myopic}).
    % while 
    % Asymmetric lookahead improves outcomes only partially (Theorems \ref{thm:asymlookahead} and \ref{thm:asymlookahead2}).
    %
    We identify several ways to mitigate the decentralization failures of type-favoring winner-take-most rules. Top-$k$ winner-selection policies interpolate between 
    % strict 
    best-type selection and broader sharing, improving decentralization while weakening performance
    % priority
    as $k$ increases (Theorems \ref{thm:fails_topk}–\ref{thm:fails_topkappa}). A certain lookahead-simulating 
    policy preserves a decentralization threshold under $\tau$-decentralization indices (Theorem~\ref{thm:asymlookahead3}).
    % , while a
    %An all-pay policy 
    % preserves participation and 
    %prevents decreases in system value but is sybil-prone
    % vulnerable
    %(Theorem \ref{thm:sybil}). %; for the lookahead-simulating policy, we also derive sufficient conditions for sybil-proofness (Theorem \ref{thm:sybil:lookahead}). 
    Broad-sharing rules never let decentralization fall, and sustain full participation whenever the largest holder's own presence does not reduce it — a condition their dynamics need not preserve. They remain sybil-prone: a split raises measured decentralization without dispersing control (Theorem \ref{thm:broadsharing} and Example \ref{ex:tail}).
\end{description}
Our results show that there is no costless way to reconcile capability and stake. Within our framework, robust decentralization requires reward rules that do more than select the most capable participant: they must also respond explicitly to concentration and, in some cases, sacrifice immediate performance priority.

Although we phrase the model in blockchain terms, the structure it isolates
recurs in repeated allocation problems more broadly: a policy repeatedly selects
a winner on the basis of an exogenous quality signal and an accumulated position,
the reward increases that position, and the resulting concentration is itself a
system-level concern---for resilience, for diversity of provision, or for
contestability. Our full-participation results speak directly to that setting.
Proposition~\ref{prop:pinvariant} and Theorem~\ref{thm:virtual} show that
selection probabilities are fixed after the first round and that initial
imbalance is never corrected; Proposition~\ref{prop:degeneracy-lambda} and
Theorem~\ref{thm:type-proportional-limit} show that letting the weight between quality and
position vary over time does not escape the resulting dichotomy. Together they
identify a limit on a common design instinct: concentration cannot be managed by
re-weighting merit against accumulated position \emph{inside the selection rule};
the reward rule must change. Settings with incumbency or reputation feedback
inherit this directly---past success confers incumbency in procurement scoring,
priority in reputation-based task routing, and track record in grant allocation.
What is specific to blockchain is the channel of Section~\ref{sec:game},
where concentration is not merely a designer's objective but is \emph{internalised}
by the agents: the decentralization of the participating set re-prices every
agent's accumulated stake through the token value, so a participation decision moves
the value of what an agent already holds.

\section{Related Work}\label{sec:related}

\emph{Reward design and participation in PoS.} Provably secure PoS begins with
Ouroboros~\cite{kiayias2017ouroboros}; \citet{saleh2021blockchain} gives an economic model in which PoS
aligns incentives without energy-intensive computation, and a line of work studies
reward sharing, pool formation and the equilibria they
induce~\cite{brunjes2020reward,gersbach2022staking,schneider2023staking}, as well as hybrid protocols combining
proof-of-work and proof-of-stake~\cite{Minotaur}. Closer to our questions,
\citet{birmpas} examine reward designs and committee sizes in governance,
\citet{chitra2021competitive} competition between staking and on-chain lending,
\citet{john2021equilibrium} equilibrium staking levels under opportunity costs, and
\citet{georganas} airdrops as strategic games. Our starting point is the framework
of \citet{chaidos2023blockchain}, which captures validator incentives under different
reward-sharing rules, and \citet{kiayias2024balancing}, which balances participation
against decentralization. Both are single-shot; we ask instead how repeated
participation decisions interact with an evolving stake distribution and an
endogenous token value. The strategy space of \citet{chaidos2023blockchain} is also richer,
allowing players to retract; we begin with the simpler binary action space.

\smallskip \noindent  \emph{Dynamics of stake shares.} \citet{rosusaleh} show that under
stake-proportional selection the investor shares form martingales converging to a
non-degenerate limiting distribution; %Observation~\ref{obs:polya} is that result in our notation, and we use it to characterise the regime into which dynamic-weight policies degenerate (Proposition~\ref{prop:44}).
Proposition~\ref{prop:pinvariant} extends the martingale to any fixed weighting of
type against accumulated stake, and Theorem~\ref{thm:virtual} draws the design
consequence: no such weighting corrects an initial imbalance. Empirically,
\citet{irresberger} find that concentration in deployed PoS systems is driven by
entry and exit rather than by disproportionate compounding of large stakes. This
is what a martingale in shares predicts, and it maps onto our two environments:
under full participation the selection rule neither concentrates nor corrects, and
it is the participation margin of Section~\ref{sec:game} that moves decentralization.

\smallskip \noindent  \emph{Measuring centralization.} Empirical work documents concentration in
deployed systems~\cite{GBERS18,grandjean2024ethereum,li2020comparison} and in pooling
behaviour~\cite{ovezik2022decentralization,motepalli2025decentralization}; \citet{motepalli2024does} propose quantitative
metrics and \citet{srivastava2024centralizationproofofstakeblockchainsgametheoretic} analyse bootstrapping protocols game-theoretically.
These motivate our treatment of decentralization as a design objective rather than
as an outcome to be measured after the fact.

\smallskip \noindent  \emph{Token value as a channel.} \citet{cong2021tokenomics} model token adoption and valuation
with feedback between usage and incentives, \citet{he2020staking} the effect of concentration
on value, and \citet{kiayiaslazos} single- versus two-token designs. This
literature justifies a value that responds to the decentralization of the
participating stake, which is what couples participation to allocation in
Section~\ref{sec:game}.

\smallskip \noindent \emph{Repeated allocation in multiagent systems.} Outside blockchain, the
structure of Section~\ref{sec:virtual} recurs wherever a mechanism selects repeatedly
and past selections shape future ones. Perpetual voting~\cite{lackner} and online
fair division~\cite{aleksandrov,freeman} study fairness across rounds of an
allocation without money; diversity constraints in multiwinner
selection~\cite{bredereck} play a role analogous to our top-$k$ eligibility rules;
and work on decisions that alter the population they act on~\cite{liu,perdomo}
shares our concern with a feedback loop between a policy and its own state. The
stochastic core of Section~\ref{sec:game} is a reinforcement process in the sense
surveyed by \citet{pemantle}.

\section{Model and Preliminaries}\label{sec:model}

\newcommand{\PL}{\mathcal{P}}
We consider a game with a set $\PL=\{1,\ldots,n\}$ of $n$ players. The game is run over a horizon of $T$ (possibly infinite) rounds. In the endogenous-participation game, at each round  $t$ every player
$i\in \PL$ has two possible actions: to \emph{Participate} or to \emph{Abstain}. 
\begin{comment}
Participating has a \emph{cost} $\tm{c_i}$ for player $i$ and round $t$. In the full-participation environment, all players participate without paying any cost.
\end{comment}
% 
At each round $t\geq 1$ the 
% system (
monetary policy
% ) 
selects a \emph{winner} among the participating players in this round, distributes rewards to
them
% players (possibly also to non-winners) 
in the form of 
% different  amounts of 
\emph{tokens}, thus changing the amount of tokens that players hold so far.  The latter determines \emph{decentralization} of the system, and thus the \emph{token value}. The choice of the winner may also depend  on the individual \emph{types} of the participants, which correspond to some intrinsic  characteristic that is important for the system (e.g., their hardware or latency). 
Specifically, we have the following key definitions and notation:
% \begin{itemize}[left=0pt,topsep=1pt,itemsep=1pt]
	% \item  
    $\boldsymbol{\tau}=(\tau_1, \ldots, \tau_n)$ is the type vector of players, with $\tau_i \geq 1$ being the type of $i$; 
	% \item
    $\tm{\sigma}$ is the token stake vector at the beginning
	of round $t$, that is, $\tm{\sigma_i}$ is the amount of tokens held by $i$;
	% \item 
    $\tm{P} \subseteq \PL$ is the set of \emph{participating} players at round $t$;
	% \item 
    $\tm{w}$ is a probability distribution over the current participating players of choosing a particular winner; 
	% \item 
    $\tm{\beta}$ is the  budget at round $t$;
	% \item 
    $\tm{b}$ is the budget vector at round $t$, which divides (part of) the budget among the participating players,  $\sum_{i \in \tm{P}} \tm{b}_i \leq \tm{\beta};$
	% \item
    $\tm{d}$ is the decentralization at round $t$;
	% \item
    $\tm{v}$ is the token value at round $t$.
% \end{itemize}

%\subsection \smallskip 
\noindent \emph{Decentralization and Monetary Policies: Definitions and Assumptions.} Before we define the utilities of the players at each round, we
% need to 
specify how the winner, the rewards,
% (budget), 
the decentralization, and token value depend on the strategies of the players at the current round. 
We stress that the strategies of the players affect both the received rewards $\tm{b_i}$ and the token value $\tm{v}$. In particular, we make the following assumptions. 
\begin{assumption}\label{assmpt:noinflation}
	The budget is constant, i.e., $\tm{\beta} = \beta$ for each
    % time step
    $t$.
\end{assumption}

Assumption~\ref{assmpt:noinflation} restricts the focus to 
policies that give constant cumulative rewards to players. Although limiting, this restriction is sufficient to highlight a fundamental tradeoff between rewards and decentralization.

\begin{assumption}\label{assmpt:onlyP}
	Decentralization is evaluated with respect to the stake of the currently participating players $\tm{P}$, that is, $\tm{d}=d(\tm{\sigma}_{\tm{P}})$, where $\tm{\sigma}_{\tm{P}}$ is the vector $\tm{\sigma}$ restricted to the participating players $\tm{P}$ only.
	The token value is of the form $\tm{v}=v(\tm{d})$ for arbitrary $v(\cdot)$ non-decreasing in %the decentralization 
	$\tm{d}$.~\footnote{The value of the token in principle also depends on the winner type (i.e., efficiency). We restrict to only decentralization since we mostly focus on policies favoring high-type agents thereby indirectly accounting for efficiency (and its effects on token values).}
\end{assumption}

Assumption \ref{assmpt:onlyP} posits that perceived decentralization, i.e., the one defined by the set of participating players, is the focus of this study. This models cases where perceived decentralization  drives market forces and the demand/supply curves for the token,
% . Perception is in general important in determining agents' behavior
as is often the case
in crypto markets \cite{cong2021tokenomics}.

We call $d(\tm{\sigma}_{\tm{P}})$ the decentralization at round $t$, and $v(d(\tm{\sigma}_{\tm{P}}))$ the \emph{system value}. Since $v$ is non-decreasing, any statement about decentralization holds %implies the corresponding statement about 
for  system value; the converse may fail, since $v$ need not be strictly increasing. We therefore state our results in terms of decentralization wherever possible.

\begin{comment}
\begin{remark}[Protocol incentives as a repeated contract for endogenous participation]
Assumption~\ref{assmpt:onlyP} 
 makes 
explicit 
that token value is forward-looking, 
% in the sense that 
i.e.,
it reacts to an evolving notion of system health (here, the decentralization of the participating stake).
% This allows us to 
We can then
interpret the protocol monetary policy as an \emph{implicit repeated contract} between the protocol and potential participants: 
% in each 
at
round $t$
the protocol specifies (i) a selection rule and (ii) a token-compensation rule, while agents choose whether to participate and incur 
% the associated
a
cost.\footnote{Under this lens, mechanisms that load incentives onto continuation values (e.g., vesting) and mechanisms that react to concentration become natural extensions of the baseline policy space considered here.}
\end{remark}
\end{comment}

Economically, $v(\cdot)$ is a \emph{common-value channel}: it converts a
system-level statistic of the participating set into a multiplier on every agent's
holdings. Monotonicity in $d$ means that agents are exposed to concentration not
only through their chance of winning, but through the value of what they already
hold; abstention therefore becomes a way of \emph{raising} the value of an existing
position rather than merely forgoing a reward. This is what couples the
participation game to the allocation dynamics, and what distinguishes our setting
from repeated allocation problems in which concentration is a designer's concern
alone.

%We do not restrict our attention to a particular decentralization measure. Instead, w
We consider any decentralization measure that satisfies the following natural properties.

\begin{assumption}\label{assmpt:dm}
	Decentralization  satisfies the following %three 
    conditions:
	\begin{enumerate}[left=0pt,topsep=1pt,itemsep=1pt]
		\item \label{dm-singleton} If there is a single participant, i.e., $\sigma$ is a singleton, then the decentralization measure is minimized;
        \item \label{dm-subset} For all $\sigma = (\sigma_1, \ldots, \sigma_n)$ and every $W \subset [n]$, there is a threshold $\theta_W < 1$ such that if $\sum_{i \in W} \sigma_i \geq \theta_W \sum_{i=1}^n \sigma_i$, then $d(\sigma) = d(\sigma_W)$, i.e., if the stake is essentially concentrated over a set $W$ of few participants, we can focus only on them, and neglect the
        others;
        % remaining ones;
		\item \label{dm-mon-remove} For all $\sigma$, if $d(\sigma) \geq d(\sigma_{-i^*})$, with $i^* \in \arg \max_j \{\sigma_j\}$, then $d(\sigma) \geq d(\sigma_{-i})$ for all $i \neq i^*$, i.e., if 
        decentralization cannot improve if we remove a participant, if it does not with the largest stake one.
        % decentralization does not improve by removing the
        % largest stake 
        % participant,
        % % with largest stake,
        % it
        % % cannot improve
        % does not
        % by removing
        % % any 
        % another participant.
	\end{enumerate}
\end{assumption}

The above assumption is satisfied by many decentralization measures usually considered in literature.
Here, we highlight this for one of these measures.
\begin{definition}[$\tau$-Decentralization Index]\label{def:Nakamoto}
	Consider w.l.o.g. a stake vector $\sigma$ ordered by weakly decreasing stake, 
	that is, $\sigma_1\geq \sigma_2\geq \cdots$. 
	The $\tau$-Decentralization Index is the minimum number of parties that collectively control more than a $\tau$ fraction of
	the total stake, i.e., $d(\sigma) = \min\{k \in \{1,\ldots,n\}| \ \sum_{i=1}^k \sigma_i> \tau\sum_i \sigma_i\}$. 
	The Nakamoto index is 
	the $\tau$-Decentralization Index for $\tau=1/2$. 
\end{definition}
We
% now verify
note
that the $\tau$-decentralization index 
satisfies Assumption~\ref{assmpt:dm}.
\begin{proposition}\label{prop:Nakamoto}
	For any value of $\tau$, the $\tau$-Decentralization Index satisfies Assumption~\ref{assmpt:dm}.
    % In particular, the Nakamoto Index does too.
\end{proposition}
This confirms that the $\tau$-Decentralization Index, including the Nakamoto Index as a special case, can be used within our framework.
We now introduce the notion of monetary policy.

\begin{definition}[Monetary Policy]
	The monetary policy $\mu = (\mu_w, \mu_b)$  determines at each round the winner (probability distribution), and the budget (expected) allocation to each participant in this round, i.e., 
	% \begin{align}
 %    \label{eq:mon-policy}
	%     \tm{w} = \mu_w(\tm{P},t) && \text{and} &&  
	% \tm{\mu_b}= \mu_b(\tm{P},t) \ . 
	% \end{align}
    $\tm{w} = \mu_w(\tm{P},t)$ and 
	$\tm{\mu_b}= \mu_b(\tm{P},t)$. 
	Hence the winner is chosen according to the probability distribution $\mu_w(\cdot)$ over $\tm{P}$. The output of %the policy
	$\mu$ can also depend on the round $t$.\footnote{The same participation set may result in different outcomes at different rounds, even when the winner is chosen deterministically.}
\end{definition}
Our definition is general and does not make any assumption on how winners and budgets are determined; our results hold at this level of generality. The budget allocated at round $t$ may depend on $\tm{w}$, the current stake profile and the set $\tm{P}$ of participants at step $t$. Here are some 
% natural
examples of monetary policies that can be 
% easily
defined in our framework, on which we will focus our attention.
\begin{description} [leftmargin=5pt,topsep=1pt,itemsep=1pt]
	\item[(Policy $\mu^\gamma$.)] This monetary policy, on 
    % receives in 
    input a vector $\mathbf{z}=(z_i)_{i \in \tm{P}}$, %where $m$ is the number of agents participating in the current round, 
    % and
    chooses as winner the participant $i$ with prob.
    % ability 
    $\frac{z_i^\gamma}{\sum_j z_j^\gamma}$ for some $\gamma \geq 1$, and allocates the budget only to the winner.\footnote{Similar formulas have been used for selecting the winner in all-pay auctions \cite{baye1996all}, and in Bitcoin PoW economies of scale for selecting the winner of next block reward \cite{arnosti2022bitcoin}.} For each $i$, $z_i$ may be a function of her type $\tau_i$ and current stake.
	\item[(Policy $\mu^{\textrm{all}}$.)] The monetary policy always chooses as winner the participant with largest type, and allocates the same share of the budget to all participants.
\end{description}
% Other instantiations of these policies will be defined below.

We
set
% will make use of the following shorthand:   
$ 
\tm{B_i}(Q) := \E_{\tm{b} \sim \mu_b(Q,t)} \left[ \tm{b_i}\right],
$ 
% which denotes 
to denote
the expected budget allocated to 
% player 
$i$ at round $t$, 
% depending on a particular 
given a
subset $Q$ of participants. 

\begin{assumption}
	\label{ass:aligned}
	Token value and budget are \emph{aligned} 
	i.e., increment in token value is more valuable for a player than new collected budget. Formally, for every $i$, every $t$, every $\tm{\sigma}$, every {set of participating players at round $t$,} $\tm{P}$, and every realization {of $\tm{b}$ drawn from $\tm{\mu_b}$}, we have that{, for every $Q \subseteq \tm{P} \setminus \{i\}$,} 
	if {$v(d(\tm{\sigma_{\tm{P}}})) < v(d(\tm{\sigma_{{Q}}}))$,
		then $v(d(\tm{\sigma_{\tm{P}}})) (\tm{\sigma_i}+\tm{b_i}) < v(d(\tm{\sigma_{{Q}}})) \tm{\sigma_i}$}.
\end{assumption}

{Assumption \ref{ass:aligned} focuses our study on mature blockchain systems, for which the existing stake is significant enough to be worth more than the new {tokens} allocated to players at each individual round.}

%\subsection
% \smallskip
\noindent \emph{Utility and Behavioral Model.}
We study myopic players for the endogenous participation setting. This form of imperfect rationality is mainly motivated by the blockchain domain. Crypto markets are often characterized by different kinds of bias and investment horizons that tend to be very short.~\footnote{Institutional investors have long-term goals and could be modeled as fully rational players; %this observation also motivates why 
the study of the repeated participation game in the classical game-theoretic setting is beyond the scope of this paper, even though it may be of theoretical interest.} 

%Myopic players are those with extremely short investment horizons. 

\begin{definition}[Myopic Players]\label{def:myopic}
	A player $i$ is \emph{myopic} if she aims to maximize her utility at each round,
    where her
    % The 
    utility 
    % of myopic player $i$
    is: $\tm{U_i}(\tm{P}) = \left( \tm{\sigma_i} + \tm{B_i}(\tm{P})
	\right)\cdot v(d(\tm{\sigma}_{\tm{P}}))$.
    %- 
	%\tm{c_i} \cdot \tm{p_i},
	%where
	%$\tm{p_i}$=1 if $i\in \tm{P}$ and $0$ otherwise.
	Myopic player $i$ participates 
    % only if participation this yields a higher utility than abstention,
    % that is, if 
    whenever doing so yields at least as much utility as abstaining,
    i.e., 
    $\tm{U_i}(\tm{P})-\tm{U_i}(\tm{P}\setminus \{i\})\geq 0$ when $i \in \tm{P}$ so that ties are broken in favor of participation.
\end{definition}

Each myopic agent will evaluate her {financial} position after her action by considering $\tm{\sigma_i} + \tm{B_i}(\tm{P})$, the expected stake 
% that she will have 
accumulated after participating or abstaining at $t$, at $v(d(\TM{\sigma_{\tm{P}}}{t}))$,  
the market value of the token when she took her decision.\footnote{Often, when Bitcoin's decentralization became critical \cite{BTC-thegardian-2014,BTC-forklog-2025}, the responsible mining pool voluntarily took action to reduce its hash power even before reaching the critical 51\% threshold. Similar concerns also affect Ethereum's PoS \cite{CNN-Ethereum-2024} and are regarded as among the major causes of a possible security breach in Hyperliquid \cite{Hyperlink-2024}, around which a rumor circulated and was followed by a 21\% token value plunge \cite{Hyperlink-plunges-2024}.  Concerns over centralized control
	also preceded a value plunge of 57\% in the Pi Network \cite{PI-plunges-2024}.
} This
% effectively
means 
the market value changes
% that the
only later
new tokens are
% first
received
% , and only later 
(i.e., for evaluation purposes, {the financial} position is closed upon the reception of the new tokens).
\section{On Policies Balancing Types and Stakes}\label{sec:virtual}
We now present our results in the full-participation environment. We study a special case of the policy $\mu^\gamma$ for 
% \begin{align}
    % \label{eq:virtual-stake}
    $\tm{z}_i=\alpha \tau_i + (1-\alpha) \tm{\sigma_i}$. 
% \end{align}
%$\tm{b}_i=\alpha \tau_i + (1-\alpha) \tm{\sigma_i}$. 
{This corresponds to the concept of \emph{virtual stake}, as in  \cite{Minotaur}, which combines type and stake to determine the winner.} 
We will see that 
% for values of
if
$\alpha$
is
not too close to $1$, this policy may 
fail to adjust unbalanced initial stakes,
% lead to problematic outcomes, 
even when dynamically adjusting $\alpha$ over time.
% : specifically, the policy may fail to adjust unbalanced initial stakes.
For this reason, 
% in the next section 
we will next focus only on policies that choose the winner based only on types. 
%Indeed, as we will show in Section~\ref{sec:virtual-stake-dynamic}, even dynamically adjusting $\alpha$ over time is not sufficient to address the issue with virtual-stake based policies, unless $\alpha$ converges to $1$, i.e., to a type-based policy.

For the sake of presentation, let us first consider the special case of $\gamma = 1$. We let $\mu^\alpha$ denote a shorthand for the policy $\mu^1$ with $\tm{\mathbf{z}}$ defined as the virtual stake vector.
% \eqref{eq:virtual-stake}.
Moreover, let $\tm{w_i} = {\tm{z_i}}/{\tm{Z}}=(\alpha \tau_i + (1-\alpha) \tm{\sigma_i})/\tm{Z},$ where $\tm{Z} = \sum_{j = 1}^n \tm{z_j}$. Thus, by definition, $\mu^\alpha$ at each round $t$ chooses participant $i$ as a winner with probability $\tm{z_i}$ and assigns a reward $\tm{\beta} = \beta$ (w.l.o.g., we normalize $\beta$ to $1$) only to the winner. 
We next prove some properties of this policy. To this aim, let us first introduce some useful notation: $T:= \sum_{j=1}^n \tau_j$ and $\tm{S}:=\sum_{j=1}^n \tm{\sigma_j}$. Observe that $\tm{Z} = \alpha T + (1-\alpha) \tm{S}$. We have the following result.

\begin{proposition}\label{prop:pinvariant}
For every player $i$, the sequence $(\tm{w_i})_{t\ge1}$ is a martingale with
respect to the stake process, i.e.,
$\mathbb{E}\big[\TM{w_i}{t+1}\mid\tm{\sigma}\big]=\tm{w_i}$ for every $t\ge1$.
Consequently, $\mathbb{E}\big[\tm{w_i}\big]=\TM{w_i}{1}$ for every $t$, and
$\tm{w_i}$ converges almost surely and in $L^1$ to a random variable
$\TM{w_i}{\infty}$ with $\mathbb{E}\big[\TM{w_i}{\infty}\big]=\TM{w_i}{1}$.
\end{proposition}

\begin{comment}
    \begin{proposition}\label{prop:pinvariant}
	For all $i$ and all $t\ge 1$ we have that $\mathbb{E}\left[\TM{w_i}{t+1}\right]=
    % \mathbb{E}\left[
    \tm{w_i}$ and,
    % \right]
    % .$
	consequently,
	$\mathbb{E}\left[\tm{w_i}\right]=
    % \mathbb{E}\left[
    \TM{w_i}{1}
    % \right]\quad\forall t
    $,
    % for all $t$,
	where
	$\TM{w_i}{1}
	=\frac{\alpha \tau_i + (1-\alpha) \TM{\sigma_i}{1}}
	{\alpha T + (1-\alpha)  \TM{S}{1}}.$ 
\end{proposition}
\end{comment}

Proposition~\ref{prop:pinvariant} essentially states that round~$1$ already fixes the proportions of \emph{all} future weights: each subsequent round merely rescales every $\tm{z_i}$ by the common factor $(1+(1-\alpha)/\tm{Z})$.

We prove that $\mu^\alpha$ has undesirable properties if $\alpha$ is away from $1$ (recall that $\mu^{\alpha}$ denotes the policy with $\gamma=1$ w.r.t. the virtual stake with parameter $\alpha$). Specifically, we prove the following result.

\begin{theorem}\label{thm:virtual}
Fix any $\delta>0$ and any $\alpha\le1-\delta$, and fix any type vector $\tau$
with $\tau_i\in[1,\tau_{\max}]$ and $n\ge2$ players; let
$i^\star\in\arg\max_i\tau_i$. Then for every $\varepsilon\in(0,1)$ there is an
initial stake profile $\TM{\sigma}{1}$ such that, under policy $\mu^\alpha$, $\lim_{t\to\infty}\tm{\sigma_{i^\star}}/\tm{S}$ exists almost surely and
satisfies
\[
\Pr\left[\lim_{t\to\infty}\frac{\tm{\sigma_{i^\star}}}{\tm{S}}\le\varepsilon\right]\ \ge\ 1-\varepsilon
\qquad\text{and}\qquad
\mathbb{E}\left[\lim_{t\to\infty}\frac{\tm{\sigma_{i^\star}}}{\tm{S}}\right]\ \le\ \varepsilon^{2}.
\]
\end{theorem}

\begin{comment}
\begin{theorem}\label{thm:virtual}
	For each $\alpha \leq 1-\delta$ for some non-negligible $\delta$, there is an initial stake profile $\TM{\sigma}{1}$, such that the largest type player in expectation collects under policy $\mu^\alpha$ at most a negligible fraction of the total stake.
\end{theorem}    
\end{comment}

Figure~\ref{fig:dynamics}(a) illustrates this: even at $\alpha=0.9$ the long-run share of the highest-type player stays far below her type share.

\begin{remark}
The restriction $\alpha\le1-\delta$ is necessary: at $\alpha=1$ the rule ignores
stake and $\tm{w_i}=\tau_i/T$ for every $t$, so $i^\star$ retains the largest
share. Theorem~\ref{thm:virtual} thus exhibits a discontinuity at $\alpha=1$:
any strictly positive structural weight on accumulated stake, however small,
allows the initial profile to render the type signal irrelevant in the limit.
Note also that $\TM{w_i}{\infty}$ is genuinely random---as in a P\'olya urn with
immigration---so limiting shares cannot be replaced by their round-one values;
only their means are pinned down.
\end{remark}

Now, we shows that the persistence of initial imbalance extends to the more general class of $\mu^\gamma$ policies,  with $\gamma>1$, that adopt  $\tm{\mathbf{z}}$ as the virtual stake vector.

\begin{theorem}\label{thm:gamma-lockin}
Fix any $\delta>0$, any $\alpha \le 1-\delta$, and any $\gamma>1$. Then, for every
$\varepsilon>0$, there exists an initial stake profile $\TM{\sigma}{1}$ such that, under
policy $\mu^\gamma$ with $z_i^{(t)}=\alpha \tau_i + (1-\alpha)\sigma_i^{(t)}$,
the player with largest type obtains an asymptotically negligible fraction of the total
stake with probability at least $1-\varepsilon$. In particular, with probability at least
$1-\varepsilon$, this player never wins any round.
\end{theorem}
\begin{comment}
\begin{proof}[Proof sketch]
The argument reduces to two players: a highest-type player $h$ and a lower-type player
$\ell$ with a sufficiently large initial stake. As long as $h$ has not yet won, the virtual stake
of $h$ stays essentially constant, while the virtual stake of $\ell$ keeps increasing because
$\ell$ receives all intermediate rewards. Therefore, the conditional probability that $h$ wins
at round $t$ decays polynomially in the initial imbalance and in $t$. 
%
More precisely, one obtains an upper bound of the form
$
q_t \le \frac{C}{(L+t-1)^\gamma},
$ 
where $L$ is the initial stake of $\ell$ and $C$ is a constant independent of $t$. Since
$\gamma>1$, the tail sum $\sum_{t\ge 1} q_t$ can be made arbitrarily small by choosing $L$
large enough. Hence, with probability at least $1-\varepsilon$, player $h$ never wins any
round. On that event, her stake remains bounded while total stake grows linearly, so her
share of the total stake converges to $0$.
\end{proof}
\end{comment}

%Theorem \ref{thm:gamma-lockin} shows that the persistence of initial imbalance
% is not specific to the linear case $\gamma=1$, but
%extends to the broader $\mu^\gamma$ family.

\subsection{Virtual Stake with Dynamic Weights} %$$-- a Refined Type--Stake Dichotomy}
\label{sec:virtual-stake-dynamic}

In this section we revisit virtual-stake policies $\mu^\alpha$ and refine our model
to allow the weight on types to vary over time.
% , so that the
% policy can react to the evolving stake distribution and the resulting
% decentralization. 
Formally, at each round $t$ the policy chooses a weight
$\alpha_t\in[0,1]$, %(possibly as a function of the past winners and current stake profile), 
and defines scores via
% \begin{equation}
	$\tm{z_i}
	=
	\alpha_t \tau_i + (1-\alpha_t) \tm{\sigma_i}$, and
	% \qquad
	$\tm{w_i}
	=
	\nicefrac{\tm{z_i}}{\tm{Z}}$.
	% \label{eq:virtual-stake-dynamic}
% \end{equation}
The stake dynamics remain unchanged. We write $\mu^{\alpha_t}$ for the
resulting family of \emph{dynamic-weight} policies. 
Allowing $\alpha_t$ to vary over time is natural from a protocol
design perspective: the policy can in principle react to the current
decentralization level, or more generally to the evolving stake
distribution, by shifting weight between type and stake.
Our next results show, however, that even with this additional
flexibility, virtual-stake
policies that only modify the \emph{selection rule} (while keeping
winner-take-all rewards) face a refined type--stake dichotomy in large
systems. 
On the one hand, if the stake component has a nontrivial 
% structural
weight,
% in~\eqref{eq:virtual-stake-dynamic}, 
then as total stake grows
the winner distribution converges to stake-proportional selection and
the influence of types on the winner probabilities vanishes.
On the other hand, if the aggregate contribution of the stake term to
the virtual-stake scores remains uniformly small over time, then the
policy behaves essentially like a type-only rule. %We next formalize these two regimes.
%\footnote{The same arguments hold if the stake term is multiplied by an arbitrary positive constant, which can be interpreted as a scaling factor between type and stake.}

%\paragraph*
% \smallskip 
\noindent \emph{Stake-only Regime.} We first show that if the stake component has a fixed structural
weight, then as the system grows large, the winner distribution
becomes arbitrarily close to stake-proportional selection, i.e.,
the contribution of types to winner selection vanishes.

\begin{proposition}\label{prop:degeneracy-lambda}
	Assume that 
    types are bounded, i.e., $\tau_i\in[0,\tau_{\mathrm{max}} ]$ for all $i$,
	%and some $\tau_{\mathrm{max}} >0$, 
    and that there exists $\delta>0$ such that
	$1-\alpha_t\ge \delta$ for all rounds $t$.
	For each round $t$ with $\tm{S}>0$, define the stake-proportional
	distribution
	$
	\tm{\tilde w_i} \;=\; \frac{\tm{\sigma_i}}{\tm{S}}$. Then, for every such round $t$,
	$
	\bigl\| \tm{w} -  \tm{\tilde w} \bigr\|_1
	\;\le\; \frac{2\alpha_t \sum_{j=1}^n \tau_j}{(1-\alpha_t) \tm{S}}
	\;\le\; \frac{2n\tau_{\mathrm{max}} }{\delta\, \tm{S}}.
	$ 
	In particular, for every $\varepsilon>0$, if
	$\tm{S}\ge \frac{2n\tau_{\mathrm{max}} }{\delta\varepsilon}$, then
	$\bigl\| \tm{w} -  \tm{\tilde w} \bigr\|_1\le \varepsilon$.
\end{proposition}

\begin{comment}
\begin{observation}[see also \cite{rosusaleh}]\label{obs:stake-martingale}
Suppose the winner is selected proportionally to current stake, i.e., 
$\alpha_t=0$ for every $t$. Then the normalised stake shares
$\tm{x_i}:=\tm{\sigma_i}/\tm{S}$ form a martingale with respect to the stake
process, namely $\mathbb{E}\big[\TM{x_i}{t+1}\mid\tm{\sigma}\big]=\tm{x_i}$ for
every player $i$. More generally, for every coalition $A\subseteq[n]$ the
normalized coalition stake $\tm{x_A}:=\sum_{i\in A}\tm{x_i}$ satisfies
$\mathbb{E}\big[\TM{x_A}{t+1}\mid\tm{\sigma}\big]=\tm{x_A}$.~\footnote{In this regime the
stake process is a classical P\'olya urn, and for positive integer initial
stakes the limiting shares are Dirichlet distributed with parameters
$\TM{\sigma}{1}$ \cite{blackwell1973}.}
\end{observation}
\end{comment}

The combination of Proposition~\ref{prop:pinvariant} 
and Proposition~\ref{prop:degeneracy-lambda} clarifies the dynamics of
the stake-only regime. Indeed, Proposition~\ref{prop:degeneracy-lambda} shows that when the
stake component retains a non-negligible structural weight, the winner
distribution becomes arbitrarily close to stake-proportional selection
as the system grows.
Proposition~\ref{prop:pinvariant}\footnote{Proposition~\ref{prop:pinvariant} assumes exact stake-proportional selection,
while Proposition~\ref{prop:degeneracy-lambda} only gives closeness to it. The gap is harmless:
the per-round drift of $\tm{\tilde w_i}$ is $\nicefrac{\tm{w_i}-\tm{\tilde w_i}}{\tm{S}+1}=O(t^{-2})$
by the bound in the proof of Theorem~\ref{thm:virtual}, hence summable, so no
systematic bias accumulates over the horizon.}, invoked with $\alpha_t=0$ for all $t$, 
%Observation~\ref{obs:stake-martingale} 
then implies that, under exact
stake-proportional selection, normalized stake shares exhibit no
expected drift toward a more balanced allocation.
Thus, once the virtual-stake policy degenerates to this regime, it no
longer provides a systematic force correcting initial stake imbalance;
rather, it behaves like a proof-of-stake lottery whose expected
evolution preserves the current stake shares.

% \smallskip
\noindent \emph{Type-only Regime.} We next consider the opposite regime, in which the contribution of the
stake term to virtual stake scores vanishes, so that
types retain a persistent influence. We show that winner probabilities become asymptotically type-proportional.

\begin{theorem}\label{thm:type-proportional-limit}
Assume that
	$
	(1-\alpha_t) \tm{S}\to 0$
	as $t\to\infty$.  
	Then, for every player $i$,
	$
	\lim_{t\to\infty} \tm{w_i}
	=
	\frac{\tau_i}{\sum_{j=1}^n \tau_j}.
	$
\end{theorem}

\begin{comment}
The two results above identify two opposite asymptotic regimes for
virtual-stake policies based only on the winner-selection rule.
If the stake component retains a non-negligible structural weight, then
as total stake grows the winner distribution becomes arbitrarily close
to stake-proportional selection, so the influence of types vanishes in
the limit.
Conversely, if the aggregate contribution of the stake term becomes
negligible, then the winner distribution converges to a lottery that is
proportional to types.
Thus, even after allowing the weight $\alpha_t$ to vary over time, purely
selection-based virtual-stake policies do not escape the underlying
tradeoff between %: asymptotically they collapse either to 
stake-based selection
%or to 
and a type-based lottery.
\end{comment}

%\input{redistribution}
\section{Endogenous Participation}\label{sec:game}
Since policies based on virtual-stake may fail to guarantee efficiency when they weigh stake too much, as shown by our results above, we here consider policies favoring better types, regardless of the stake. Clearly, in the setting of full participation, this would lead the player with highest type to accumulate all available tokens, with severe effects on decentralization. For this reason, we will now move from 
the full-participation benchmark to the endogenous-participation environment, where the policy affects not only stake dynamics but also the set of active participants through the feedback loop on the token value.

All results in the main body of the paper assume myopic players. Appendix \ref{app:lookahead} introduces a different behavioral model; results there state this difference explicitly.

%We assume for convenience  that $\tm{c_i} = 0$ for every $i$ and every $t$. In any case, it is not hard to see that our results can be easily extended to non-zero costs, as long as these satisfy appropriate assumptions (cf. Observation~\ref{obs:costs}).

\subsection{Equilibria Characterization}
Henceforth,  we assume w.l.o.g., that in each round $t$ players are sorted according to their stake, i.e., $\tm{\sigma_1} \geq \tm{\sigma_2} \geq \cdots \geq \tm{\sigma_n}$. In what follows, we let $P^{\geq i}$ denote the set of participants where $i$ has the largest stake and all players with smaller stake participate, that is, $P^{\geq i} = \{i, i+1, \ldots, n\}$.  

Observe that single-shot games play a central role in  our behavioral model. Accordingly, we herein focus on the equilibria of single-shot games. To this end, we introduce the notion of harmful stake profiles for players and observe a useful monotonicity property of decentralization measures in harmful stake profiles.

\newcommand{\expb}[2]{\E_{#1,#2}[\TM{b_i}{#1}]}

\begin{definition}[Harmful stake profile]
  We say that the stake profile $\tm{\sigma}$ is \emph{harmful} for player $i$ in the participation set $P \ni i$ if $\tm{U_i}(P) < \tm{U_i}(P \setminus \{i\})$.
\end{definition}

%We first observe that no stake profile is harmful for the last two agents in $P^{\geq n-1}$ and $P^{\geq n}$.
%\begin{proposition}\label{prop:noharmforlast2}
%    A stake profile $\tm{\sigma}$ is never harmful for player $n$ in $P^{\geq n}$ nor for player $n-1$ in
    % respectively  and 
%    $P^{\geq n-1}$.
%\end{proposition}

\begin{lemma}\label{lem:harmful-d}
Let $\tm{\sigma_i}>0$. If $\tm{\sigma}$ is harmful for player $i$ in $P$, then
$d(\sigma_{P\setminus\{i\}})>d(\sigma_P)$.
\end{lemma}

% \smallskip \noindent \textbf{Myopic players.} 
We next %prove the following result, which 
characterize the equilibria of these monetary policies. %for myopic players.
\begin{lemma}
\label{lem:equilibria}
For every initial stake profile and every monetary policy $\mu$, %if players are myopic then for every time step $t$ 
the single-shot game at each round $t$ admits a
unique equilibrium, and this equilibrium is a suffix of players with respect to
the stake order. It can be computed in polynomial time.
\end{lemma}
\begin{proof}[Proof idea.]
The key observation is that profitable participation is monotone along the stake order. If, for some player $i$, the suffix $P^{\ge i}$ is preferable to abstaining from it, then every player with smaller stake also prefers to participate once $P^{\ge i}$ forms. Conversely, if participation is harmful for $i$, then either abstention propagates upward until a player is reached for whom participation is not harmful, or a recovery winner is found that supports a larger suffix. This creates a unique cutoff player and therefore a unique equilibrium suffix. The full proof is in appendix. 
\end{proof}
%The proof of this lemma formalizes why it is sufficient to focus only on suffixes of players, as equilibria of each stage game. Intuitively, if agent $i$ finds the suffix $P^{\geq i}$ more profitable than $P^{\geq i} \setminus \{i\}$, then the decentralization measure is large enough to make it also profitable for agents with stake lower than $i$.

%\paragraph*
\smallskip \noindent \textbf{Players with asymmetric lookahead.}
We next characterize the equilibrium for players with asymmetric lookahead. 

\begin{lemma}
\label{lem:equilibria:lookahead}
For every initial stake profile $\TM{\sigma}{1}$, and every monetary policy $\mu$, if players have asymmetric lookahead, then for every time step $t$ there is a unique equilibrium given by a suffix of players. 
\end{lemma}
The lemma shows how to define recovery plans for abstaining players in a way that the properties of equilibria for myopic players are preserved for these more sophisticated agents. However, these recovery plans come at a cost, as we cannot guarantee that equilibria can be computed in polynomial time in this case.

\subsection{On Policies Favoring Better Types}\label{sec:typeonly}
We now look at monetary policies $\mu^{\gamma}$, with 
% $\gamma \rightarrow \infty$ and 
$\tm{z_i}=\tau_i$, i.e., those choosing the winner proportionally to the types of the participating players.
% with overwhelming probability.
We denote these policies by $\mu_\tau^\gamma$. 

Throughout, we say that an event holds with overwhelming probability if its failure probability tends to 0 in the relevant asymptotic parameter: the deviation $\varepsilon$ of the winner distribution from deterministic class selection in Theorems \ref{thm:fail_generic}, \ref{thm:fails_myopic} and \ref{thm:asymlookahead2}, and the number of rounds in Theorems \ref{thm:fails_topk} and \ref{thm:fails_topkappa}. 

% \smallskip 
\noindent \emph{Generic $\mu_\tau^\gamma$
% , myopic players, 
and $\tau$-decentralization index.}
We next show that for these policies, the decentralization eventually enters a well-defined low-decentralization region and returns to it infinitely often. Specifically, for every $\gamma > 0$, we can partition the set of players in sets $W_1, \ldots, W_k$ such that if at least one player in $W_i$ participates and no player in $\bigcup_{j < i} W_{j}$ participates, then the winner would be in $W_i$ with probability at least $1-\varepsilon$, for some $\varepsilon < \nicefrac{1}{2}$. For example, if we have five players with types $(10,10,1,1,1)$, then for the policy $\mu_\tau^\gamma$, we can define the following sets: $W_1 = \{1, 2\}$ and $W_2 = \{3, 4, 5\}$.
Let $w_i$ denote the maximum possible decentralization %system value 
within $W_i$. E.g., if the decentralization %system value 
is exactly the Nakamoto index, then $w_i = \left\lfloor\nicefrac{|W_i|}{2} + 1\right\rfloor$, and thus in the instance given above we would have $w_1 = 2$ and $w_2 = 2$. Let $w = \max_k w_k$.
% and $k_w = \min \{k : w_k = w\}$, i.e., $k_w$ is the first index such that the maximum possible value within $W_{k_w}$ is $w$. 
%We have the following theorem.
\begin{theorem}
\label{thm:fail_generic}
    If decentralization is measured through the $\tau$\nobreakdash-\hspace{0pt}decentralization index, then, for every $\gamma > 1$, the decentralization eventually drops to at most $w$ from any starting value, and attains $w$ infinitely often. Finally, for $\gamma \rightarrow \infty$ (i.e., $\varepsilon = 0$), once decentralization reaches $w$, it never exceeds $w$ again.
\end{theorem}
\begin{proof}[Proof sketch]
We look at the \emph{lowest-priority} class $W_j$ whose players still win
infinitely often. Once all lower-priority classes stop winning, only players in $W_j$ and
higher-priority classes can keep increasing their stake. Eventually, the stake of the players
in $W_j$ who continue to win dominates that of the remaining players in lower classes. At
that point, whenever a winner is selected from $W_j$, the participation set is effectively
confined to a subset of $W_j$, up to players whose contribution can be neglected by
Assumption~3. Hence the decentralization %system value 
at those rounds is at most $w_j\le w$.

It remains to control what happens in rounds in which players from higher-priority classes
win. Here the properties of the $\tau$-decentralization index imply that giving additional
stake to players already responsible for concentration cannot move the decentralization %system 
above the
relevant bound, while giving stake to players inside the active class can improve
decentralization only up to the maximum value attainable within that class. Therefore the
process cannot escape the low-decentralization region indefinitely.
\end{proof}

\noindent \emph{Policy $\mu^\star$.}
% and myopic players.}
We next turn our attention to very special cases of $\mu_\tau^\gamma$ in which each set $W_i$ is a singleton and $\varepsilon$ is negligibly small. Note that this requires that $\gamma \rightarrow \infty$ and either players' types are distinct or a fixed tie-breaking rule is enforced among players with same type (by modifying the probability distribution accordingly). We will denote these policies with $\mu^*$.

Note that in this case $w$ is the value extractable from a singleton, which, by Condition 1 of Assumption \ref{assmpt:dm}, is the minimum. Hence, Theorem~\ref{thm:fail_generic} implies that in these cases the system will eventually reach the minimum $\tau$-decentralization index. The following example highlights the centralization forces induced by $\mu^*$.
\begin{example}%[Myopic Participation with Nakamoto Index]
\label{ex:myopic participation}
Three players with types $\tau_1 > \tau_2 > \tau_3$ and initial stake 1. 
Assume that the token's value equals the Nakamoto index of the participating set and that 
% one
a
new token is awarded to the highest-type participant (i.e., $\varepsilon = 0$). The dynamics unfold as follows: %(see Table~\ref{tab:myopic} for a summary):
\begin{itemize} [left=0pt,topsep=0pt,itemsep=0pt]
  \item \textbf{Rounds 1--2:} All players participate; the Nakamoto index is~2, and Player~1 %, having the highest type, 
  receives the new token.
  \item \textbf{Round 3:} The initial stake profile is harmful for Player~1, since participation would reduce the index to~1 and guarantee utility at most $4\cdot 1=4$. A recovery winner (see proof of Lemma \ref{lem:equilibria}) exists: if Players~2 and~3 participate, the index remains~2, guaranteeing utility $3\cdot 2=6$. Hence, only Players~2 and~3 join, and Player~2 receives the token.
  \item \textbf{Round 4:} Player~1 participates, all players join, the index is~2, and Player~1 receives the token.
  \item \textbf{Round 5 and beyond:} All continue to participate, since for any set of participants the index would drop to~1. Centralization thus occurs, and Player~1 secures all subsequent tokens.
\end{itemize}
\begin{comment}
\begin{table}[t]
    \centering
    \rowcolors{2}{gray!10}{white} 
    \resizebox{\columnwidth}{!}{
\begin{tabular}{@{}c c c c c c@{}}
\toprule
\textbf{Round} & \textbf{Initial Stake} & \textbf{Participants} & \textbf{Nakamoto Index} & \textbf{Winner} & \textbf{Final Stake} \\
\midrule
1 & (1, 1, 1) & All three & 2 & Player 1 & (2, 1, 1) \\
2 & (2, 1, 1) & All three & 2 & Player 1 & (3, 1, 1) \\
3 & (3, 1, 1) & P2, P3     & 2 & Player 2 & (3, 2, 1) \\
4 & (3, 2, 1) & All three & 2 & Player 1 & (4, 2, 1) \\
5+ & (4{+}, 2, 1) & All three & 1 & Player 1 & (5{+}, 2, 1) \\
\bottomrule
\end{tabular}}
    \caption{Evolution of stakes, participation, Nakamoto index, and reward allocation with myopic players}
    \label{tab:myopic}
\end{table}    
\end{comment}
\end{example}
% We next prove that t
The result of Theorem~\ref{thm:fail_generic} can be strengthened for
% policy
$\mu^*$.
To this aim, we define a failing monetary policy.
% as follows.
\begin{definition}[Failing monetary policy]
  We say that a monetary policy \emph{fails} if there is an instance of the game (i.e., initial stake, player types, %for each player, 
  and %a 
  budget %-generation function 
  $\beta$) for which there is a round $t_0$ such that the following {is true with overwhelming probability}. For every $t \geq t_0$, the %system value %(i.e., 
  decentralization %measure) 
  at round $t$ is at its minimum value. 
\end{definition}
We next show that $\mu^*$ fails for every decentralization measure satisfying Assumption~\ref{assmpt:dm}.
% , and not just for $\tau$-decentralization index.
Moreover, once the system has reached its minimum value, it 
% does not recover to 
never takes
a larger value except with negligible probability, even
% when
for
$\varepsilon > 0$.
%We have the following result.
\begin{theorem}
\label{thm:fails_myopic}
%Monetary policy 
$\mu^*$ fails, both for $\varepsilon = 0$ and for $\varepsilon > 0$. 
%if players are myopic. The claim holds, except with negligible probability, even if $\varepsilon > 0$.
\end{theorem}
Interestingly, the proof establishes
% a stronger statement, 
% % and it is
% namely,
that eventually all players will participate at every time step.
% We also stress that Theorem~\ref{thm:fails_myopic} does not require the equilibrium characterization of  Lemma~\ref{lem:equilibria}.

\begin{comment}
\begin{observation}
\label{obs:costs}
The claim above will hold unchanged even if we assume participation costs larger than $0$, as long $\tm{\beta}$ is large enough (e.g., large and constant transaction fees collected, chiming with our mature blockchain assumption). 
This way, the monetary policy can always allocate shares of $\tm{\beta}$ to players, so that each receives in expectation a share that is at least as large as her cost. So the properties of the policy, namely that the expected winner share is larger that the other players' by a non-negligible amount and that in expectation the share received by non-winner players do not increase decentralization, hold for (expected) \emph{relative} shares, i.e., shares minus cost.
\end{observation}
\end{comment}

\subsection{Top-$k$ eligibility}\label{sec:topk}
We now focus on policy $\mu^k$ that, at each time step $t$, sorts participants in non-increasing order of type with a fixed tie-breaking rule, and sets $z_i = 1$ if $i \leq k$ and $z_i = 0$, otherwise. That is, the policy $\mu^k$ at each time step chooses the winner uniformly at random among the top-$k$ participants with respect to type.

Note that this policy does not belong to the class of policies covered by Theorem~\ref{thm:fail_generic}. Still, we next prove that a similar result holds even for $\mu^k$.
Indeed, it is not hard to check that $\mu^*$, whenever $\varepsilon = 0$, can be seen as a special case of $\mu^k$ for $k=1$. Hence, in this case, by Theorem~\ref{thm:fails_myopic} we know that the system will eventually reach the minimum value and this value will never increase. We show that the argument used above to prove this result can be extended even to larger values of $k$ in order to guarantee that the system value will eventually reach a threshold $\vartheta$ and with high probability never surpass it. The threshold $\vartheta$ is, similarly to what happens with $\mu_\tau^\gamma$, the maximum value that the system can achieve when restricted to only $k$ players. 
Specifically, let $\vartheta_s \;=\; \max_{\sigma_1,\dots,\sigma_k}
\Big\{\, d(\sigma_1,\dots,\sigma_k)\ \Big|\ \textstyle\sum_{i=1}^{k}\sigma_i=s,\ \sigma_i\ge0 \,\Big\}$.
Then we take $\vartheta=\liminf_{s\to\infty}\vartheta_s$; small totals are
immaterial, since the total stake grows without bound along the dynamics and
$\vartheta_1=1$ for every $k$. For example, if decentralization is measured
through the Nakamoto index then $\vartheta=\lceil k/2\rceil$.~\footnote{Let $\sigma_1\ge\cdots\ge\sigma_k\ge0$ be integers with
$\sum_i\sigma_i=s$. For even $k$, $\sum_{i\le k/2}\sigma_i\ge s/2$, with equality
only if all the $\sigma_i$ coincide, i.e., only if $s$ is a multiple of $k$; hence
$\vartheta_s=\tfrac{k}{2}+1$ when $k$ divides $s$, and $\vartheta_s=\tfrac{k}{2}$
otherwise. For odd $k$, $\lceil k/2\rceil=\lfloor k/2\rfloor+1$ is the largest
value $d$ can take on $k$ players, and a near-balanced profile attains it for all
large $s$. Either way $\liminf_{s}\vartheta_s=\lceil k/2\rceil$.}

%: for $k=2$ the maximum achievable decentralization is $1$ whenever $s$ is odd, while for $k=3$ and every $s\ge2$ one can always arrange $\sigma_1,\sigma_2,\sigma_3$ so that the decentralization is $2$.
%For example, if decentralization is measured through Nakamoto index, then $\vartheta = \left\lceil\frac{k}{2}\right\rceil$: indeed, if $k=2$ the maximum possible decentralization when $s$ is odd is 1, while if $k = 3$, for every $s$ it is always possible to arrange $\sigma_1$, $\sigma_2$, and $\sigma_3$ so that the decentralization is $2$.
%Changed def in d units
%$\vartheta = \min_{s > 0} \max_{\sigma_1, \ldots, \sigma_k} \left\{v(d(\sigma_1, \ldots, \sigma_k)) \mid \sum_{i=1}^k \sigma_i = s, \sigma_i \geq 0\right\}$.
\begin{theorem}
\label{thm:fails_topk}
Assume that decentralization is measured through the $\tau$-decentralization index. % and players are myopic. 
Monetary policy $\mu^k$ will eventually lead the system to have value $\vartheta$, and with overwhelming %high 
probability it will have this value for most of the successive steps.
\end{theorem}
We stress that the proof shows that the probability that the system value has a value different from $\vartheta$ decreases as the number of steps increases, and thus it may become overwhelming only after a large number of steps.

Let us now consider the deterministic variant of $\mu^k$. We define the policy $\mu^k_{\mathrm{det}}$ that, at each time step $t$, chooses as winner the participant with the lowest stake among the top-$k$ participants with respect to type. Note that this still will fit in the $\mu^\gamma$ policy framework by requiring $\gamma = \infty$, and by setting $z$ as for $\mu^k$, except that for the one participant with lowest stake among the ones with $z = 1$, we increase her value by a tiny
 amount.
%$\delta$.
% 
For this policy, we
% can 
strengthen Theorem~\ref{thm:fails_topk}, by showing that
% the policy
it
will eventually lead the system value to be \emph{surely} at least $\vartheta$ for \emph{all} the next steps.
\begin{theorem}
\label{thm:fails_topkappa}
% Monetary policy 
$\mu^k_{\mathrm{det}}$ will eventually lead the system to 
have value at least $\vartheta$ and keep this value for all the following steps.
\end{theorem}

Figure~\ref{fig:dynamics}(b,c) compares the policies of this section along the equilibrium path, and shows that the decentralization gained by raising $k$ is paid for in efficiency.

\subsection{Lookahead-Simulating Policies}

Top-$k$ eligibility keeps decentralization high when $k$ is large, but at a cost in
performance. We now ask whether a policy can hold decentralization above a
threshold while selecting closer to the best type. Our route is indirect: we first
ask what would happen if agents were less myopic, and then construct a policy that
reproduces that behaviour against agents who are not.

Say that a player has \emph{asymmetric lookahead} if she is myopic when deciding
whether to participate, but evaluates abstention through the equilibrium path that
follows from dropping out---so a temporary exit is judged by the system value it
eventually restores, rather than by the reward it immediately forgoes. Such players
can partly escape Theorem~\ref{thm:fails_myopic}. For an initial profile
$\TM{\sigma}{1}$ and a player $i$, let $\theta_i$ be the largest system value below
which $\TM{\sigma}{1}$ is harmful for $i$, and set $\theta=\min_i\theta_i$. Run
against players with asymmetric lookahead, $\mu^{\star}$ keeps the system value
above $\theta$, though it does not prevent the system from reaching its minimum
along the way.

This suggests a policy for \emph{myopic} agents. Let $\mu^{\ell}$ select, at each
round $t$, the winner that $\mu^{\star}$ would have chosen at round $t+1$ had the
agents had asymmetric lookahead: simulating the lookahead inside the policy
substitutes for it in the agents.

\begin{theorem}\label{thm:asymlookahead3}
Let $d$ be measured via a $\tau$-decentralization index. For every
initial stake profile $\TM{\sigma}{1}$ that is not harmful for the agent with
largest type and whose system value is at least $\theta$, $\mu^{\ell}$ keeps the
system value at least $\theta$ even when players are myopic. Moreover the winner
sequence it produces coincides asymptotically with that of $\mu^{\star}$, up to the
one-round shift.
\end{theorem}

Figure~\ref{fig:dynamics}(b,c) shows how $\mu^{\ell}$ compares with the other policies defined in the paper. Appendix~\ref{app:lookahead} gives the behavioral model formally, the equilibrium
characterization it induces, the two results that motivate $\mu^{\ell}$, and worked
dynamics under both regimes.

\subsection{Broad Sharing and Sybil-Proofness}\label{sec:broad}

We now consider the policy $\mu^{\mathrm{all}}$, which chooses as winner the
participant with largest type but shares the budget equally among all
participants. Under this policy selection and reward are decoupled: being chosen
determines who executes the task, and hence system performance, but confers no
additional reward. What participation buys a player is therefore a share of the
budget, weighed against any change in decentralization her presence causes. The
following notion isolates the first of these.

\begin{definition}\label{def:broadsharing}
A monetary policy is \emph{broad-sharing} if every participant receives a strictly
positive share of the budget, i.e.\ $\tm{B_i}(P)>0$ for every $i\in P$ and every
round $t$.
\end{definition}

\begin{theorem}\label{thm:broadsharing}
Let $\mu$ be broad-sharing, let decentralization be measured through a
$\tau$-decentralization index, and let $i^*$ be a player with largest stake at
round $t$. If $d(\tm{\sigma})\ge d(\tm{\sigma}_{-i^*})$, then full participation is
an equilibrium at round $t$. If moreover $\mu=\mu^{\mathrm{all}}$, the winner is
the player with largest type and $d(\TM{\sigma}{t+1})\ge d(\tm{\sigma})$.
\end{theorem}
\begin{proof}[Proof sketch]
Condition~3 of Assumption~\ref{assmpt:dm} reduces $n$ participation constraints to
one: if removing the largest holder does not raise decentralization, then removing
any player does not either. Broad sharing then makes every constraint slack, since
each participant receives a strictly positive share of the budget while
$v(d(\cdot))$ does not fall.
\end{proof}

The condition of Theorem~\ref{thm:broadsharing} is not preserved by the dynamics
that $\mu^{\mathrm{all}}$ itself induces.

\begin{example}\label{ex:tail}
Let $n=4$, $\tau=\nicefrac12$, $\beta=1$, $v$ the identity, and
$\TM{\sigma}{1}=(8,3,1,1)$ with $\tau_1>\tau_2>\tau_3>\tau_4$. The condition of
Theorem~\ref{thm:broadsharing} holds for four consecutive rounds, so all players
participate and each gains $\nicefrac14$ per round, and the profile reaches
$(9,4,2,2)$. There $d(\TM{\sigma}{5})=1$ while $d(\TM{\sigma}{5}_{-1})=2$: player~1
obtains $(9+\nicefrac14)\cdot 1$ by participating against $9\cdot 2$ by abstaining,
and abstains. Her own share never falls below one half, so $d(\tm{\sigma})$ is
unchanged throughout; what changes is the tail, which equal sharing drives from
$(3,1,1)$, where one player dominates, to $(4,2,2)$, where none does.\footnote{In this instance, player 1 re-enters at the following round, once the tail has grown
further. When the condition of Theorem~\ref{thm:broadsharing} fails at the outset,
however, the exclusion can persist: with $\TM{\sigma}{1}=(50,2,2)$ the largest
holder abstains for $46$ consecutive rounds before participation becomes
profitable again.} 
%Broad sharing erodes the very condition under which it sustains participation. 
\end{example}

Example~\ref{ex:tail} also exposes the drawback of $\mu^{\mathrm{all}}$. At the
profile $(9,4,2,2)$, player~1 may instead split into two sybils holding
$\nicefrac92$ each, assigning type $\tau_1-\zeta$ to one and $\zeta$ to the
other. The participating profile becomes $(\nicefrac92,\nicefrac92,4,2,2)$, so
$d=2$; she collects two of five budget shares; and her first sybil still has the
largest type, so she is still selected. Her utility is
$(9+\nicefrac25)\cdot2=18.8$, against $18$ for abstaining and $9.25$ for
participating honestly. Sybils are thus doubly profitable under broad sharing:
they multiply the reward, and they \emph{manufacture} the decentralization that
prices the attacker's own holdings. We next observe that some implementations of
$\mu^{\ell}$ are much less susceptible to sybil incentives, 
%Still, this policy has one severe drawback: it trivially incentivizes players to create sybils in order to increase the received reward. We instead observe that some implementations of  $\mu^\ell$ are much less susceptible to sybil incentives, 
especially in settings with large competition. The exact statement is given in Appendix \ref{sec:muell:sybilpf} as it requires formally introducing the model to study sybils. 
\begin{comment}
To exemplify the result consider Example~\ref{ex:myopic participation}. 
We here just note that for %the policy 
$\mu^\ell$ to be sybil-proof there, it is sufficient that the players have a difference in type value less than $1$. Indeed, to make the sybil with largest type not harmful,
% it must be the case that
we need to reduce her stake by at least $1$. If the difference between the types of the players is at most $1$, then the first player has no incentive to create sybils.    
\end{comment}

%Taken together, Theorems~\ref{thm:asymlookahead3}--\ref{thm:sybil:lookahead} highlight two distinct ways to correct the failures of strict performance-first design. The lookahead-simulating policy $\mu^\ell$ preserves a decentralization threshold while asymptotically tracking the winner sequence induced by $\mu^*$ under asymmetric lookahead. Although $\mu^\ell$ can still admit sybil incentives, its vulnerability is limited and structurally characterizable: Theorem \ref{thm:sybil:lookahead} gives sufficient conditions under which the policy is sybil-proof. By contrast, the all-pay policy $\mu^{all}$ preserves participation and prevents decreases in system value in a very direct way, but its broad reward sharing makes sybil splitting immediately profitable. Thus, $\mu^\ell$ offers a more selective and more sybil-robust correction than $\mu^{all}$, whereas $\mu^{all}$ achieves stronger participation guarantees at the cost of much weaker resistance to identity splitting.
\begin{figure}[t]
  \centering
  \includegraphics[scale=1]{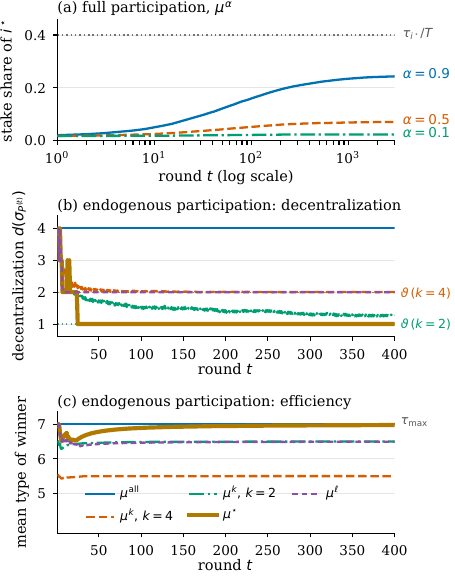}
  \caption{\textbf{(a)} Full participation under $\mu^{\alpha}$: $n=4$, types
$(4,3,2,1)$, initial stakes $(1,20,20,20)$; mean stake share of the highest-type
player $i^{\star}$ over 400 runs. Each curve converges to $\TM{w_{i^\star}}{1}$
(Proposition~\ref{prop:pinvariant}) and never approaches the type share
$\tau_{i^\star}/T=0.40$, however $\alpha$ is chosen
(Theorem~\ref{thm:virtual}).
\textbf{(b), (c)} Endogenous participation with myopic agents: $n=7$, types
$(7,\dots,1)$, equal initial stakes, $v$ the identity, $d$ the Nakamoto index;
$\mu^{k}$ averaged over 400 runs. Panel (b) gives the decentralization of the
participating set, with $\vartheta=\lceil k/2\rceil$ marked for each top-$k$
policy; panel (c) the cumulative mean type of the winner. In these experiments, raising $k$ from 2 to 4
gains a point of decentralization and costs a type unit of efficiency, whereas
$\mu^{\ell}$ attains the decentralization of $k=4$ at the efficiency of $k=2$
(curves coincide in (b) and (c) respectively); $\mu^{\mathrm{all}}$ attains both at
the price of sybil-proofness (Example~\ref{ex:tail}).}
  \label{fig:dynamics}
\end{figure}

\begin{table}[t]
\centering\small
\setlength{\tabcolsep}{2pt}
\begin{tabular}{@{}llll@{}}
\toprule
Policy & Winner & Decentralization & Result \\
\midrule
$\mu^{\alpha},\mu^{\gamma}$ & virtual-stake lottery & imbalance persists & Thm~\ref{thm:virtual}--\ref{thm:gamma-lockin} \\
$\mu^{\alpha_t}$ & dynamic-weight lottery & collapses to $\mu^\alpha$ or $\mu^\gamma$
  & Prop.~\ref{prop:degeneracy-lambda}, Thm~\ref{thm:type-proportional-limit} \\
$\mu^{\gamma}_{\tau}$ & best type w.p.\ $1-\varepsilon$ & $\le w$ inf.\ often & Thm~\ref{thm:fail_generic} \\
$\mu^{*}$ & best type & reaches minimum & Thm~\ref{thm:fails_myopic} \\
$\mu^{k}$ & uniform in top-$k$ & reaches $\vartheta$ & Thm~\ref{thm:fails_topk} \\
$\mu^{k}_{\det}$ & lowest stake in top-$k$ & $\ge \vartheta$ & Thm~\ref{thm:fails_topkappa} \\
$\mu^{\ell}$ & as $\mu^{*}$, shifted & $\ge\theta$ & Thm~\ref{thm:asymlookahead3} \\
$\mu^{\mathrm{all}}$ & best type & never decreases & Thm~\ref{thm:broadsharing} \\
\bottomrule
\end{tabular}
\caption{Policies studied, the winner each selects, and the decentralization
guarantee it provides. Only $\mu^{\mathrm{all}}$ and $\mu^{\ell}$ sustain full
participation, and only under the conditions of Theorem~\ref{thm:broadsharing} and
Theorem~\ref{thm:asymlookahead3} respectively; $\mu^{\mathrm{all}}$ is sybil-prone.
(Example~\ref{ex:tail}).
}
\label{tab:summary}
\end{table}

\section{Conclusions}\label{sec:conclusions}
Blockchain reward design must balance capability-sensitive performance incentives against stake-sensitive decentralization. This paper studies this tradeoff in %two environments: 
a full-participation benchmark, which isolates the effect of winner selection, and an endogenous-participation game, 
% in which
where
token value feeds 
% back into 
the
participation incentives. 

The main message is that winner selection alone may not be enough, whenever our conditions are met. Under full participation, fixed virtual-stake rules exhibit strong path dependence, while dynamic-weight rules asymptotically collapse either to stake\nobreakdash-\hspace{0pt}proportional or to type-proportional selection. Under endogenous participation, sufficiently type-favoring winner-selection policies push the system toward low decentralization, and in the extreme winner-take-most regime can drive it to minimum decentralization under myopic behavior, and even under asymmetric lookahead the system may still dip to minimum decentralization along the path.  
At the same time, the analysis points to the kinds of interventions that can work. Policies that may preserve participation or decentralization do so only by departing from strict performance-first design. They work by relaxing strict best-type winner concentration, and this can happen through continuation-sensitive selection, top-$k$ eligibility, or broader reward sharing. The lesson for protocol design is therefore not simply to tune the selection weights, but to let the monetary policy react directly to concentration through the reward rule. Table \ref{tab:summary} summarizes the findings for all the policies considered, whilst Figure \ref{fig:dynamics}  depicts dynamics and trade-offs. 

Beyond blockchain, the full-participation analysis identifies a limit on a common
design instinct: when a repeated allocation policy scores agents on a mix of merit
and accumulated position, no choice of that mix---fixed or adaptive---both rewards
merit and arrests concentration. Settings in which past success confers incumbency
or reputation, but in which no common systemic value re-prices existing positions,
inherit this limitation directly. The endogenous-value channel of
Section~\ref{sec:game} is more demanding, since it requires a systemic value
to which all agents' holdings are exposed; we do not claim that it transfers
without further structure.

Our model is intentionally stylized. Extending it to endogenous token supply, participation costs, delegation and pooling, identity or entry costs, or multi-dimensional notions of capability would be natural next steps. Another important direction is to connect strategic abstention more directly to security by studying how changes in participating stake alter attack thresholds.

\bibliography{references}

\onecolumn
\appendix
\section{Proofs}

\subsection{Proof of Proposition~\ref{prop:Nakamoto}}
\begin{proof}
It is immediate to check that Condition~\ref{dm-singleton} and \ref{dm-subset} 
of Assumption~\ref{assmpt:dm} holds.
As for Condition~\ref{dm-mon-remove}, suppose, by contradiction, that for some $i \neq i^*$, it holds that $k = d(\sigma) < d(\sigma_{-i})$. Since $k = d(\sigma)$ there is a coalition of agents $C$ of size $k$ such that
\begin{equation}
\label{eq:nakamoto_k}
    % \sum_{j \in C} \sigma_j > \frac{1}{2} \sum_{j} \sigma_j,
    \sum_{j \in C} \sigma_j > \tau \sum_{j} \sigma_j,
\end{equation}
but, since $d(\sigma_{-i})>k$, then for every coalition $C'$ of size $k$ that does not include $i$ it holds that
\begin{equation}
\label{eq:nakamoto_notk}
 % \sum_{j \in C'} \sigma_j \leq \frac{1}{2} \sum_{j \neq i} \sigma_j.   
 \sum_{j \in C'} \sigma_j \leq \tau \sum_{j \neq i} \sigma_j.
\end{equation}
This in turn implies that $i \in C$, otherwise 
% $\sum_{j \in C} \sigma_j \leq \frac{1}{2} \sum_{j \neq i} \sigma_j \leq \frac{1}{2} \sum_{j} \sigma_j$,
$\sum_{j \in C} \sigma_j \leq \tau \sum_{j \neq i} \sigma_j \leq \tau \sum_{j} \sigma_j$,
contradicting \eqref{eq:nakamoto_k}. But then since $i\in C$, we have that
% $\sigma_\ell < \frac{1}{2}\sigma_i$
$\sigma_\ell < (1-\tau)\sigma_i < \sigma_i$
for every $\ell \notin C$, otherwise it follows from \eqref{eq:nakamoto_notk} that
\begin{align*}
    & \sum_{j \in C} \sigma_j = \sum_{j \in C \setminus \{i\}} \sigma_j + \sigma_i = \sum_{j \in C \cup \{\ell\} \setminus \{i\}} \sigma_j + \sigma_i - \sigma_\ell\\
    % \leq & \frac{1}{2} \sum_{j \neq i} \sigma_j  + \sigma_i - \sigma_\ell = \frac{1}{2} \sum_{j} \sigma_j + \frac{\sigma_i}{2} - \sigma_\ell \leq \frac{1}{2} \sum_{j} \sigma_j,
    \leq & \tau \sum_{j \neq i} \sigma_j  + \sigma_i - \sigma_\ell = \tau \sum_{j} \sigma_j + (1-\tau)\sigma_i - \sigma_\ell \leq \tau \sum_{j} \sigma_j,
\end{align*}
which contradicts \eqref{eq:nakamoto_k}. Since $\sigma_{i^*} \geq \sigma_i$, we can conclude that $i^* \in C$. Consider then the coalition $C^*$ of size $k$ that contains all players in $C$ except $i^*$, and one player $\ell^{\max} \notin C$ with maximum stake. Note that every other coalition $C'$ of size $k$ that does not contain $i^*$ {can only have lower total stake than $C^*$, i.e., $\sum_{j \in C'} \sigma_j \leq \sum_{j \in C^*} \sigma_j$ since all the players outside $C$ (and thus all the players outside $C^*$ but $\ell^{\max}$) have stake lower than $\sigma_i$.} Thus, since $\sigma_i \leq \sigma_{i^*}$, we have that for every $C'$ of size $k$ that does not contain $i^*$, it follows from \eqref{eq:nakamoto_notk} that
\begin{align*}
 & \sum_{j \in C'} \sigma_j \leq \sum_{j \in C^*} \sigma_j = \sum_{j \in C} \sigma_j + \sigma_\ell - \sigma_{i^*} = \sum_{j \in C \cup \{ {\ell^{\max}} \} \setminus \{i\}} \sigma_j + \sigma_i - \sigma_{i^*}\\
 % \leq & \frac{1}{2} \sum_{j \neq i} \sigma_j  + \sigma_i - \sigma_{i^*} = \frac{1}{2} \left(\sum_{j \neq i^*} \sigma_j + \sigma_i - \sigma_{i^*}\right) \leq \frac{1}{2} \sum_{j \neq i^*} \sigma_j.\\
 \leq & {\tau} \sum_{j \neq i} \sigma_j  + \sigma_i - \sigma_{i^*} = {\tau} \sum_{j \neq i^*} \sigma_j + {(1-\tau) \left(\sigma_i - \sigma_{i^*}\right)} \leq {\tau} \sum_{j \neq i^*} \sigma_j.
\end{align*}
Hence, $d(\sigma_{-i^*})$ must be at least $k+1$, contradicting the hypothesis that $k = d(\sigma) \geq d(\sigma_{-i^*})$.
\end{proof}

\subsection{Proof of Proposition~\ref{prop:pinvariant}}
\begin{proof}
Recall that we normalize $\beta=1$. Exactly one token is awarded per round, so
$\TM{S}{t+1}=\tm{S}+1$; hence $\tm{Z}=\alpha T+(1-\alpha)\tm{S}$ is deterministic
and $\TM{Z}{t+1}=\tm{Z}+(1-\alpha)$. Given $\tm{\sigma}$, the stake of $i$
increases by one with probability $\tm{w_i}$ and is unchanged otherwise, so,
using $\tm{z_i}=\tm{w_i}\tm{Z}$,
\[
\mathbb{E}\big[\TM{z_i}{t+1}\mid\tm{\sigma}\big]
=\alpha\tau_i+(1-\alpha)\big(\tm{\sigma_i}+\tm{w_i}\big)
=\tm{z_i}+(1-\alpha)\tm{w_i}
=\tm{w_i}\big(\tm{Z}+(1-\alpha)\big)
=\tm{w_i}\,\TM{Z}{t+1}.
\]
Dividing by the deterministic $\TM{Z}{t+1}$ gives the stated identity; since the
winner distribution depends on the past only through $\tm{\sigma}$, conditioning
on the entire history yields the same value. Finally $\tm{w_i}\in[0,1]$, so the
martingale converges almost surely by the martingale convergence theorem
\cite{williams1991}; boundedness upgrades this to convergence in $L^1$, whence
$\mathbb{E}[\TM{w_i}{\infty}]=\TM{w_i}{1}$.
\end{proof}
\begin{comment}
    \begin{proof} Assume w.l.o.g. that $\beta = 1$.
First, observe that, since the reward assigned by the policy is $1$ for the winner and $0$ for every remaining player, then
\(
  \TM{S}{t+1}=\tm{S} + \sum_i \tm{w_i} = \tm{S} + 1.
\)
Hence
\[\TM{W}{t+1} = \alpha T + (1-\alpha) \TM{S}{t+1} = \alpha T + (1-\alpha) (\tm{S}+1) = \tm{W}+(1-\alpha) = \tm{W} \left(1 + \frac{1-\alpha}{\tm{W}}\right)\].
Now, let us insert $\tm{w_i} = \frac{\tm{p_i}}{\tm{W}}$ into $\tm{p_i} = \alpha \tau_i + (1-\alpha) \tm{\sigma_i}$ one step ahead:
\[
  \mathbb{E}\left[\TM{p_i}{t+1}\right] 
    =\alpha \tau_i + (1-\alpha)\mathbb{E}\left[\TM{\sigma_i}{t+1}\right]
    =\alpha \tau_i + (1-\alpha)\bigl(\tm{\sigma_i}+\tm{w_i}\bigr)
    =\tm{p_i} + (1-\alpha)\tm{w_i}
    = \tm{w_i}\tm{W}\left(1 - \frac{1-\alpha}{\tm{W}}\right).
\]
Hence, \(\mathbb{E}\left[\TM{w_i}{t+1}\right] = \frac{\mathbb{E}\left[\TM{p_i}{t+1}\right]}{\TM{W}{t+1}} = \tm{w_i}\), as desired.
If we consider the expectation over all possible randomization along $t$ steps, then a simple induction on $t$ completes the proof.
\end{proof}
\end{comment}

\subsection{Proof of Theorem~\ref{thm:virtual}}
\begin{proof}
With\ $\beta=1$, we have that $\tm{S}=\TM{S}{1}+t-1$ and
$\tm{Z}=\alpha T+(1-\alpha)\tm{S}$, where $T=\sum_j\tau_j$. Since $\tau_i\le T$
and $\tm{\sigma_i}\le\tm{S}$, both $\tm{S}\tau_i$ and $\tm{\sigma_i}T$ lie in
$[0,\tm{S}T]$, so $\big|\tm{S}\tau_i-\tm{\sigma_i}T\big|\le\tm{S}T$. Hence, for
every player $i$,
\[
\left|\tm{w_i}-\frac{\tm{\sigma_i}}{\tm{S}}\right|
=\frac{\alpha\big|\tm{S}\tau_i-\tm{\sigma_i}T\big|}{\tm{S}\,\tm{Z}}
\ \le\ \frac{\alpha T}{\tm{Z}}
\ \le\ \frac{\alpha T}{\delta\,\tm{S}}\ \xrightarrow[t\to\infty]{}\ 0,
\]
where the last inequality uses $1-\alpha\ge\delta$. By
Proposition~\ref{prop:pinvariant}, $\tm{w_{i^\star}}$ converges almost surely to
$\TM{w_{i^\star}}{\infty}$ with
$\mathbb{E}[\TM{w_{i^\star}}{\infty}]=\TM{w_{i^\star}}{1}$; as we have just proved,
$\tm{\sigma_{i^\star}}/\tm{S}$ has the same almost sure limit. It therefore
suffices to exhibit an initial profile with
$\TM{w_{i^\star}}{1}\le\varepsilon^{2}$.

Set $\TM{\sigma_{i^\star}}{1}=1$ and $\TM{\sigma_j}{1}=L$ for every
$j\ne i^\star$, where $L$ is fixed below. Since $1-\alpha\ge\delta$,
\[
\TM{w_{i^\star}}{1}
=\frac{\alpha\tau_{i^\star}+(1-\alpha)}
      {\alpha T+(1-\alpha)\big(1+(n-1)L\big)}
\ \le\ \frac{\alpha\tau_{\max}+(1-\alpha)}{(1-\alpha)(n-1)L}
\ \le\ \frac{\tau_{\max}+1}{\delta L}.
\]
Choosing $L\ge(\tau_{\max}+1)/(\delta\varepsilon^{2})$ therefore gives
$\mathbb{E}[\TM{w_{i^\star}}{\infty}]=\TM{w_{i^\star}}{1}\le\varepsilon^{2}$,
which is the second claim. The first follows by Markov's inequality:
$\Pr[\TM{w_{i^\star}}{\infty}>\varepsilon]\le\varepsilon^{2}/\varepsilon=\varepsilon$.
\end{proof}

\begin{comment}
\begin{proof}
Assume w.l.o.g. that $\beta = 1$.
Since we are assuming full participation at each round, and selection probabilities are in expectation invariant according to Proposition~\ref{prop:pinvariant}, then the fraction of times player $i$ has been selected as winner by the policy $\mu^\alpha$ over $t$ rounds tends to $\TM{w_i}{1}$ as $t$ goes to infinity. Hence, in expectation, the stake accumulated by $i$ tends to $\TM{\sigma_i}{1} + t\TM{w_i}{1}$, where $\TM{w_i}{1} = \frac{\alpha \tau_i+(1-\alpha)\TM{\sigma_i}{1}}{\TM{W}{1}}$, as $t$ goes to infinity.

Let $i^*$ be the player with largest type, and consider $\TM{\sigma}{1}$ such that $\TM{\sigma_{i^*}}{1} = 1$, and there is $i \neq i^*$ such that $\TM{\sigma_i}{1} \geq \frac{\alpha \left(\tau_{i^*} - \tau_i\right)}{(1-\alpha)} + M$, for some large $M$. Then the stake of $i$ will tend to $\TM{\sigma_i}{1} + t \frac{\alpha \tau_{i^*} + (1-\alpha) M}{\TM{W}{1}}$, while the stake of $i^*$ tends to $1 + t \frac{\alpha \tau_{i^*} + \beta}{\TM{W}{1}}$, that is negligibly smaller than the stake of $i^*$. \end{proof}
\end{comment}

\subsection{Proof of Theorem \ref{thm:gamma-lockin}}
\begin{proof}
It is sufficient to prove the claim for two players. We assume also that $\beta = 1$. It is immediate to check that the proof still holds for every $\beta > 0$.
Let $h$ be the player with largest
type and let $\ell$ be a player with smaller type, i.e., $\tau_h > \tau_\ell$. Consider
the initial stake profile
\[
\TM{\sigma_h}{1} = 1
\qquad\text{and}\qquad
\TM{\sigma_\ell}{1} = L,
\]
where $L$ is a parameter to be fixed later.

Assume that player $h$ has not won in the first $t-1$ rounds. Then $\sigma_h^{(t)} = 1$ and $\sigma_\ell^{(t)} = L+t-1$, since exactly one token is awarded at each round and, along this history, all previous
rewards went to player $\ell$. Therefore,
\[
\tm{z_h} = \alpha \tau_h + (1-\alpha),
\qquad
\tm{z_\ell} = \alpha \tau_\ell + (1-\alpha)(L+t-1).
\]

Let $q_t$ be the conditional probability that player $h$ wins at round $t$, given that
she has not won before round $t$. By definition of $\mu^\gamma$,
\[
q_t
=
\frac{\left(\tm{z_h}\right)^\gamma}
{\left(\tm{z_h}\right)^\gamma + \left(\tm{z_\ell}\right)^\gamma}
\le
\frac{\left(\tm{z_h}\right)^\gamma}{\left(\tm{z_\ell}\right)^\gamma}.
\]
Since $\alpha \le 1-\delta$, we have 
\[
\tm{z_\ell}
\ge
(1-\alpha)(L+t-1)
\ge
\delta(L+t-1).
\]
Hence,
\[
q_t
\le
\frac{(\alpha \tau_h + (1-\alpha))^\gamma}{\delta^\gamma (L+t-1)^\gamma}
=
\frac{C}{(L+t-1)^\gamma},
\quad
\text{where}
\quad
C := \frac{(\alpha \tau_h + (1-\alpha))^\gamma}{\delta^\gamma}.
\]

Let $E$ be the event that player $h$ wins at least one round. By a union bound over the
first winning time,
\[
\Pr[E]
\le
\sum_{t\ge 1} q_t
\le
C\sum_{t\ge 1} \frac{1}{(L+t-1)^\gamma}
=
C\sum_{s\ge L} s^{-\gamma}.
\]
Since $\gamma>1$, the series $\sum_{s\ge L} s^{-\gamma}$ converges to $0$ as $L\to\infty$.
Therefore, for every $\varepsilon>0$, we can choose $L$ sufficiently large so that
\[
C\sum_{s\ge L} s^{-\gamma} \le \varepsilon.
\]
It follows that
\[
\Pr[E] \le \varepsilon,
\qquad
\Pr[E^c] \ge 1-\varepsilon.
\]

On the event $E^c$, player $h$ never wins, and thus
$
\tm{\sigma_h} = 1$
{for every } $t$. 
Moreover, since one token is allocated at each round,
$
\tm{S} = \TM{S}{1} + t - 1 = (L+1)+t-1 = L+t$. 
Hence,
\[
\frac{\sigma_h^{(t)}}{\tm{S}} = \frac{1}{L+t} \xrightarrow[t\to\infty]{} 0.
\]
Therefore, with probability at least $1-\varepsilon$, the player with largest type
obtains an asymptotically negligible fraction of the total stake. 
\end{proof}

\subsection{Proof of Proposition \ref{prop:degeneracy-lambda}}

\begin{proof}
	Fix a round $t$ with $\tm{S}>0$.
	For brevity write $\alpha=\alpha_t$ and omit the explicit dependence
	on $t$.
	Define
	\[
	A_i := (1-\alpha) \sigma_i,\qquad
	B_i := \alpha \tau_i,
	\]
	so that $z_i = A_i + B_i$.
	Let
	\[
	A := \sum_{j=1}^n A_j = (1-\alpha) \tm{S},\qquad
	B := \sum_{j=1}^n B_j = \alpha\sum_{j=1}^n \tau_j,
	\]
	and note that $A>0$ since $1-\alpha\ge\delta>0$ and $\tm{S}>0$; moreover, $A_j\geq 0$, $B_j \geq 0$ and $B\geq0$.
	The winner probabilities can be written as
	\[
	w_i = \frac{A_i+B_i}{A+B},
	\qquad
	\tilde w_i = \frac{A_i}{A} = \frac{\sigma_i}{\tm{S}}.
	\]
	
	We first derive a pointwise bound on $|w_i-\tilde w_i|$.
	Using the identity
	\[
	\frac{A_i+B_i}{A+B} - \frac{A_i}{A}
	= \frac{(A_i+B_i)A - A_i(A+B)}{A(A+B)}
	= \frac{B_iA - A_iB}{A(A+B)},
	\]
	we obtain
	\[
	|w_i-\tilde w_i|
	=
	\left|
	\frac{B_i}{A+B}-\frac{A_iB}{A(A+B)}
	\right|
	\le
	\frac{B_i}{A+B}+\frac{A_iB}{A(A+B)}.
	\]
	
	Summing this inequality over all $i$ yields
	\[
	\sum_{i=1}^n |w_i - \tilde w_i|
	\;\le\; \frac{\sum_{i=1}^n B_i}{A+B}
	+ \frac{\sum_{i=1}^n A_i}{A}\cdot \frac{B}{A+B}
	= \frac{2B}{A+B}.
	\]
	Since $A,B\ge 0$, we have $A+B\ge A$, and then
	\[
	\bigl\| w - \tilde w \bigr\|_1
	\;\le\; \frac{2B}{A}
	= \frac{2\alpha\sum_{j=1}^n \tau_j}{(1-\alpha) \tm{S}}.
	\]
	By the assumption $\tau_j\in[0,\tau_{\mathrm{max}} ]$ for all $j$, we have
	$\sum_{j=1}^n\tau_j\le n\tau_{\mathrm{max}} $, and by assumption
	$1-\alpha\ge\delta$, so
	\[
	\bigl\| w - \tilde w \bigr\|_1
	\;\le\; \frac{2\alpha n\tau_{\mathrm{max}} }{(1-\alpha) \tm{S}}
	\;\le\; \frac{2n\tau_{\mathrm{max}} }{\delta\, \tm{S}}.
	\]
	This proves the desired bound.
	The final claim follows directly: given $\varepsilon>0$, if
	$\tm{S}\ge \frac{2n\tau_{\mathrm{max}} }{\delta\varepsilon}$, then
	$\bigl\| \tm{w} -  \tm{\tilde w} \bigr\|_1\le \varepsilon$.
\end{proof}

\begin{comment}
\paragraph*{Proof of Observation \ref{obs:stake-martingale}}
\begin{proof}
	Fix a player $i$ and let w.l.o.g., $\beta = 1$.
	Let $\tm{X_i}\in\{0,1\}$ denote the indicator that player $i$ wins at
	round $t$.
	By assumption, the winner receives one token and all other players
	receive zero, so
	\[
	\TM{\sigma_i}{t+1}=\tm{\sigma_i}+\tm{X_i}.
	\]
	Since the winner is selected proportionally to current stake, we have
	\[
	\mathbb{E}[\tm{X_i}\mid \tm{S}] = \frac{\tm{\sigma_i}}{\tm{S}}=\tm{x_i}.
	\]
	Moreover, exactly one token is allocated in total, hence
	\[
	\TM{S}{t+1}=\tm{S}+1.
	\]
	Therefore
	\[
	\TM{x_i}{t+1}=\frac{\TM{\sigma_i}{t+1}}{\TM{S}{t+1}}
	=\frac{\tm{\sigma_i}+\tm{X_i}}{\tm{S}+1}.
	\]
	Taking conditional expectation given $\tm{S}$, we obtain
	\[
	\mathbb{E}[\TM{x_i}{t+1}\mid \tm{S}]
	=
	\frac{\tm{\sigma_i}+\mathbb{E}[\tm{X_i}\mid \tm{S}]}{\tm{S}+1}
	=
	\frac{\tm{\sigma_i}+\tm{\sigma_i}/\tm{S}}{\tm{S}+1}
	=
	\frac{\tm{\sigma_i}\frac{\tm{S}+1}{\tm{S}}}{\tm{S}+1}
	=
	\frac{\tm{\sigma_i}}{\tm{S}}
	=
	\tm{x_i}.
	\]
	This proves that $\{\tm{x_i}\}_t$ is a martingale.
	
	Now fix any coalition $A\subseteq N$.
	Since
	\[
	\tm{x_A}=\sum_{i\in A} \tm{x_i},
	\]
	linearity of conditional expectation gives
	\[
	\mathbb{E}[\TM{x_A}{t+1}\mid \tm{S}]
	=
	\sum_{i\in A}\mathbb{E}[\TM{x_i}{t+1}\mid \tm{S}]
	=
	\sum_{i\in A} \tm{x_i}
	=
	\tm{x_A}.
	\]
	Hence $\{\tm{x_A}\}_t$ is also a martingale.
\end{proof}
\end{comment}

\subsection{Proof of Theorem \ref{thm:type-proportional-limit}}
\begin{proof}
	Fix a player $i$.
	For every round $t$, define
	\[
	R_i(t):=(1-\alpha_t) \tm{\sigma_i},
	\qquad
	R(t):=(1-\alpha_t) \tm{S}=\sum_{j=1}^n R_j(t).
	\]
	Then
	\[
	\tm{z_i}=\alpha_t\tau_i+R_i(t),
	\qquad
	\sum_{j=1}^n \tm{z_j}=\alpha_t\sum_{j=1}^n \tau_j + R(t).
	\]
	Hence
	\[
	\tm{w_i}
	=
	\frac{\alpha_t\tau_i+R_i(t)}
	{\alpha_t\sum_{j=1}^n \tau_j + R(t)}.
	\]
	Since $\tm{\sigma_i}\le \tm{S}$ for every $i$, we have
	\[
	0\le R_i(t)\le R(t).
	\]
	By assumption, $R(t)\to 0$ as $t\to\infty$, and therefore also
	$R_i(t)\to 0$ for every $i$.
	
	We now compare $\tm{w_i}$ with
	\[
	\frac{\tau_i}{\sum_{j=1}^n \tau_j}.
	\]
	Write $\tau_{\mathrm{sum}} :=\sum_{j=1}^n \tau_j$.
	Then
	\[
	\tm{w_i}-\frac{\tau_i}{\tau_{\mathrm{sum}} }
	=
	\frac{\tau_{\mathrm{sum}} (\alpha_t\tau_i+R_i(t))-\tau_i(\alpha_t\tau_{\mathrm{sum}} +R(t))}
	{\tau_{\mathrm{sum}} (\alpha_t\tau_{\mathrm{sum}} +R(t))}.
	=
	\frac{\tau_{\mathrm{sum}}  R_i(t)-\tau_i R(t)}
	{\tau_{\mathrm{sum}} (\alpha_t\tau_{\mathrm{sum}} +R(t))}.
	\]
	Since $\tm{S}\ge \TM{S}{1}>0$ for every $t$, the assumption
	$(1-\alpha_t)\tm{S}\to 0$ implies $1-\alpha_t\to 0$, and hence
	$\alpha_t\to 1$. So for $t$ large enough, there exists  $\bar\alpha>0$ such that
	$\alpha_t\ge \bar\alpha$. Taking absolute values for $t$ large enough, we thus obtain
	\[
	\left|\tm{w_i}-\frac{\tau_i}{\tau_{\mathrm{sum}} }\right|
	\le
	\frac{\tau_{\mathrm{sum}}  |R_i(t)|+\tau_i |R(t)|}
	{\tau_{\mathrm{sum}} (\alpha_t\tau_{\mathrm{sum}} +R(t))} 	\le
	\frac{\tau_{\mathrm{sum}}  R_i(t)+\tau_i R(t)}
	{\tau_{\mathrm{sum}}^2 \bar\alpha }
	=
	\frac{R_i(t)}{\bar\alpha\tau_{\mathrm{sum}} }
	+
	\frac{\tau_i R(t)}{\bar\alpha\tau_{\mathrm{sum}} ^2}, \]
	where in the last inequality we have used $R_i(t)\ge 0$, $R(t)\ge 0$, and $\alpha_t\ge \bar\alpha$.

	As observed above, both $R_i(t)\to 0$ and $R(t)\to 0$ as
	$t\to\infty$.
	Hence the right-hand side tends to $0$, and thus
	\[
	\lim_{t\to\infty} \tm{w_i}
	=
	\frac{\tau_i}{\tau_{\mathrm{sum}} }
	=
	\frac{\tau_i}{\sum_{j=1}^n \tau_j}.
	\]
	This proves the claim.
\end{proof}

%\subsection{Proof of Proposition~\ref{prop:noharmforlast2}}
%\begin{proof}
%    It follows immediately from the definition that $\tm{\sigma}$ is never harmful for player $n$ in $P^{\geq n}$. 
%    By contradiction, suppose $\tm{\sigma}$ is harmful for player $n-1$ in $P^{\geq n-1}$.  By Lemma \ref{lem:harmful-d}, $d(\sigma_{P^{\geq n}})< d(\sigma_{P^{\geq n-1}})$, contradicting Condition \ref{dm-singleton} in Assumption \ref{assmpt:dm} since $d(\sigma_{P^{\geq n}})$ is at its minimum.
%\end{proof}

\subsection{Proof of Lemma \ref{lem:harmful-d}}
\begin{proof}
Since $\tm{B_i}(P)\ge0$, harmfulness gives
$\tm{\sigma_i}\,v(d(\sigma_P))\le(\tm{\sigma_i}+\tm{B_i}(P))\,v(d(\sigma_P))
<\tm{\sigma_i}\,v(d(\sigma_{P\setminus\{i\}}))$, so
$v(d(\sigma_P))<v(d(\sigma_{P\setminus\{i\}}))$, and the claim follows since $v$
is non-decreasing.
\end{proof}

\subsection{Proof of Lemma~\ref{lem:equilibria}}
\renewcommand{\bot}{\mathbf{par}}
 %Let us start by proving the results for the case of myopic players \newww{($k=0$)}.
 \begin{proof}
Consider the following procedure that assigns to each player $i$ for which $\tm{\sigma}$ is harmful in $P^{\geq i}$ a label in $\{i+1, \ldots, n\} \cup \{\bot\}$ as follows: consider players in non-decreasing order of stake; for each player $i$ such that $\tm{\sigma}$ is harmful in $P^{\geq i}$, let $r>i$ be the first index such that either $\tm{\sigma}$ is not harmful for $r$  in $P^{\geq r} $or the label of $r$ is $\bot$; if
 \begin{equation}
 \label{eq:recovery}
  v(d(\sigma_{P^{\geq i}})) \left(\tm{\sigma_i} +
  \tm{B_i}(P^{\geq i}) 
  %\sum_{j \in P^{\geq i}} \tm{w_j}(P^{\geq i}) \tm{b_i}(j)
  \right) 
  < v(d(\sigma_{P^{\geq r}})) \tm{\sigma_i},
 \end{equation}
 then assign label $r$ to $i$, otherwise assign label $\bot$. So player $i$ for which $\tm{\sigma}$ is harmful in $P^{\geq i}$ has  label $\bot$ if she prefers $P^{\geq i}$ over $P^{\geq r}$ 
 or $r$ otherwise; in this case  %If player $i$ has been assigned a label different from $\bot$, then 
 we say that $i$ has \emph{recovery winner} $r$. Intuitively, player $i$ with label $\bot$ will want to participate even though $\tm{\sigma}$ can be harmful for her because either the suffixes $P^{\geq l}$ for which \eqref{eq:recovery} is true are not an equilibrium due to the fact that $\tm{\sigma}$ is also harmful for $l$ or there is player $r$ who will participate and suffix $P^{\geq r}$ is not better than $P^{\geq i}$ to $i$.
 %$P^{\geq l}$ that may be an equilibrium are worse for her than $P^{\geq i}$; in fact, either \eqref{eq:recovery} is false for $P^{\geq r}$ (and $r$ has label $\bot$) or $\tm{\sigma}$ is also harmful for all the players $l>i$ for which \eqref{eq:recovery} is true. 
 In contrast, player $i$ with recovery winner $r$ prefers not to participate because she is guaranteed that each player between $i+1$ and $r-1$ does not participate (because participating is harmful for them and they have a recovery winner), $r$ will participate (because participating is not harmful or she has no recovery winner), and $i$ receives higher expected utility if the set of participants is $P^{\geq r}$ vis-\`a-vis $P^{\geq i}$. Note that the labeling procedure defined above can be run in polynomial time.

 Consider now the following procedure, which given the current stake profile $\tm{\sigma}$ and policy $\mu$ returns a set of participants as follows: considering players in non-increasing order of stake, return $P^{\geq i}$ for the first $i$ in this order for which %and for each player $i$, 
 $\tm{\sigma}$ is not harmful for $i$ in $P^{\geq i}$ or $i$ does not have a recovery winner. Note that also this procedure can be run in polynomial time.

 We claim that the set of participants $P$ returned by this procedure is in equilibrium for myopic players.
 %To this end, we first observe 
 The proof that for each $j \in P$, it is an equilibrium to participate relies on proving the following:
 \begin{equation}
 \label{eq:condition_eq}
  v(d(\sigma_P)) \left(\tm{\sigma_{j^*}} +
  \tm{B_{j^*}}(P)
  %\sum_{\ell \in P} \tm{w_\ell}(P) \tm{b}_{j^*}(\ell)
  \right) \geq v(d(\sigma_{P \setminus \{j^*\}})) \tm{\sigma_{j^*}},
 \end{equation}
 where $j^* = \min \{i \in P\}$. Equation~\eqref{eq:condition_eq} directly proves the claim for $j^*$.
 As for $j \in P$, we have $ j>j^*$, and \eqref{eq:condition_eq} implies, by Assumption~\ref{ass:aligned}, that $v(d(\sigma_P)) \geq v(d(\sigma_{P \setminus \{j^*\}}))$, and thus, in turn, that $d(\sigma_P) \geq d(\sigma_{P \setminus \{j^*\}})$; hence, by the properties of decentralization measure (Condition \ref{dm-mon-remove} of Assumption \ref{assmpt:dm}) $d(\sigma_{P}) \geq d(\sigma_{P \setminus \{j\}})$ for all $j>j^*$, and thus $$v(d(\sigma_{P})) \left(\tm{\sigma_j} + 
 \tm{B_j}(P)
 %\sum_{x \in P} \tm{w_x}(P) \tm{b_j}(x)
 \right) \geq v(d(\sigma_{P\setminus\{j\}})) \tm{\sigma_j},$$
 meaning that no such $j>j^*$ has an incentive to abstain and drop out of $P$. 
 
 To prove %that $v(d(\sigma_P)) \geq v(d(\sigma_{P \setminus \{j^*\}}))$ 
 Equation \eqref{eq:condition_eq}, we distinguish %three
 two cases. 
 
 Assume first that $\tm{\sigma}$ is not harmful for $j^*$ in $P^{\geq j^*}$. Then Equation~\eqref{eq:condition_eq} follows immediately from the definition of harmfulness.
 
 Suppose now that $\tm{\sigma}$ is harmful for $j^*$ in $P^{\geq j^*}$. By definition of the procedure returning $P$, this means that 
 $j^*$ has no recovery winner. We now consider two sub-cases that distinguish whether $j^*+1$ could have been a ``potential'' recovery winner for $j^*$ or not.
 
 If $\tm{\sigma}$  is not harmful for $j^*+1$ in $P^{\geq j^*+1}$ or she has no recovery winner, then \eqref{eq:recovery} does not hold for $i = j^*$ and $r = j^*+1$ since $j^*+1$ is not a recovery winner for $j^*$. This immediately implies that \eqref{eq:condition_eq} holds, since $\sigma_{P^{\geq j^*}} = \sigma_{P}$ and $\sigma_{P^{\geq j^*+1}} = \sigma_{P \setminus \{j^*\}}$.

 Consider now the case that $\tm{\sigma}$ is harmful for both $j^*$ and $j^*+1$, $j^*$ has no recovery winner, while 
 $j^*+1$ has recovery winner $r^*$. Since $r^*$ is a recovery winner for $j^*+1$, then it is either the case that $\tm{\sigma}$  is not harmful for $r^*$ in $P^{\geq r^*}$ or that $r^*$ has label $\bot$, and Equation~\eqref{eq:recovery} holds for $i = j^*+1$ and $r = r^*$, which in turn implies that 
 \begin{equation}
     \label{eq:jstar1vskstar}
     v(d(\sigma_{P^{\geq j^*+1}})) < v(d(\sigma_{P^{\geq r^*}})).
 \end{equation}
 Conversely, since $r^*$ is not a recovery winner for $j^*$, then it must be the case that Equation~\eqref{eq:recovery} does not hold for $i = j^*$ and $r = r^*$. By Assumption~\ref{ass:aligned}, this implies that
 \begin{equation}
     \label{eq:jstarvskstar}
     v(d(\sigma_{P^{\geq j^*}})) \geq v(d(\sigma_{P^{\geq r^*}})).
 \end{equation}
 Equation~\eqref{eq:condition_eq} then follows by joining Equation~\eqref{eq:jstarvskstar} with Equation~\eqref{eq:jstar1vskstar}, and observing, as above, that $\sigma_{P^{\geq j^*}} = \sigma_{P}$ and $\sigma_{P^{j^*+1}} = \sigma_{P \setminus \{j^*\}}$.   

To complete the proof that $P$ is an equilibrium for myopic players, we next prove that for $j < j^*$, it is convenient not to participate. Indeed, since they have been removed by the procedure above, it means that participation is harmful for all of them and they have a recovery winner. Moreover, it is immediate to check that, according to the procedure of label assignment, for all these players, the recovery winner must be {$j^*$}, and thus, by definition of recovery winner, all these players prefer the set of participants $P$ to their own participation.

 We conclude the proof  by observing that 
 the arguments above show that no other equilibrium exists since players break ties in favor of participation (cf. Definition \ref{def:myopic}). \end{proof}

\subsection{Proof of Theorem~\ref{thm:fail_generic}}
In order to prove Theorem~\ref{thm:fail_generic}, we first state some properties of the $\tau$-decentralization index.
\begin{lemma}
\label{lem:tau_prop}
Suppose that decentralization is measured through $\tau$-decentralization index. Then the following properties hold:
\begin{itemize}
    \item \label{dm-no-increase1} For every $\sigma = (\sigma_1, \ldots, \sigma_n)$ and every $W \subset [n]$ such that $d(\sigma) < d(\sigma_W)$, it holds that $d(\sigma') \leq d(\sigma)$ for every $\sigma'$ such that $\sigma'_i > \sigma_i$ for some $i \notin W$, and $\sigma'_j = \sigma_j$ for every $j \neq i$. In other words, if there is a set of players, namely the ones in $[n] \setminus W$, causing large centralization, then by increasing the stake of one of these players cannot decrease centralization.
    \item \label{dm-no-increase2} For every $\sigma = (\sigma_1, \ldots, \sigma_n)$ and every $W \subset [n]$ such that $d(\sigma) < d(\sigma_W)$, it holds that $d(\sigma') \leq d(\sigma_W)$ for every $\sigma'$ such that $\sigma'_i > \sigma_i$ for some $i \in W$, and $\sigma'_j = \sigma_j$ for every $j \neq i$. In other words, if there is a set of players, namely the ones in $[n] \setminus W$, causing large centralization, then by increasing the stake of one of the remaining players can decrease decentralization but still it remains at most $d(\sigma_W)$.
    \item \label{dm-no-increase3} For every $W \subseteq [n]$, $\max_{\sigma} d(\sigma) \leq \max_{\sigma_W} \sigma_W$, i.e., the maximum decentralization cannot increase if it is evaluated on a smaller set of agents.
\end{itemize}
\end{lemma}
\begin{proof}
  Observe that let $d = d(\sigma)$ and $d_W = d(\sigma_W)$, and $b = \max_i \left(\sigma'_i - \sigma_i\right)$. Since $d(\sigma) = d$, then there exist a set $S \subseteq [n]$ such that $|S| = d$ and $\frac{\sum_{i \in S} \sigma_i}{\sum_{i \in [n]} \sigma_i} > \tau$. Note that, by integrality of $d$ and $d_W$, it must be that $d_W \geq d+1$. Let us first consider the setting of Condition~\ref{dm-no-increase1}, i.e., that $\sigma'$ is achieved from $\sigma$ by increasing of an amount $b$ the stake of some $i \notin W$. Then
\[
\frac{\sum_{i \in S} \sigma'_i}{\sum_{i \in [n]} \sigma'_i} = \frac{b + \sum_{i \in S} \sigma_i}{b + \sum_{i \in [n]} \sigma_i} \geq \frac{\sum_{i \in S} \sigma_i}{\sum_{i \in [n]} \sigma_i} > \tau,
\]
from which we achieve that $d(\sigma')$ cannot be larger than $d$ as desired.

Consider now the setting of Condition~\ref{dm-no-increase2}, i.e., that $\sigma'$ is achieved from $\sigma$ by increasing of an amount $b$ the stake of some $w\in W$. Then, we consider the set $S' = S \cup \{w\}$ of size $d+1$, and observe that
\[
\frac{\sigma'_w + \sum_{i \in S} \sigma'_i}{\sum_{i \in [n]} \sigma'_i} = \frac{b + \sigma_w + \sum_{i \in S} \sigma_i}{b + \sum_{i \in [n]} \sigma_i} \geq \frac{b + \sum_{i \in S} \sigma_i}{b + \sum_{i \in [n]} \sigma_i} \geq \frac{\sum_{i \in S} \sigma_i}{\sum_{i \in [n]} \sigma_i} > \tau,
\]
from which we achieve that $d(\sigma')$ cannot be larger than $d+1 \leq d_W$ as desired.

Finally, Condition~\ref{dm-no-increase3} trivially holds.
\end{proof}

\begin{proof}[Proof of Theorem~\ref{thm:fail_generic}]
Let us first assume that the policy selects the winner in $W_i$ \emph{surely} if at least one player of $W_i$ participates and no player in $\bigcup_{j < i} W_j$ participates, i.e., $\varepsilon = 0$.
 %, and provides a reward to this player that is larger of at least a constant value $\delta$ with respect the (possibly null) reward assigned to other participants.

Let $j$ be the largest index such that at least one player in $W_j$ that wins infinitely often.
We first observe that there is $t_0$ such that for every $t \geq t_0$, if at step $t$ a player in $W_j$ wins then the decentralization %value of the system 
in this step is at most $w_j \leq w$. Indeed, this is obvious if $j = k$, because, by definition of $W_k$ and the fact that a player in $W_k$ wins, we have that $\tm{P} \subseteq W_k$, and, by Condition~\ref{dm-no-increase3} of Lemma~\ref{lem:tau_prop}, $d(\tm{\sigma}_{\tm{P}})$ cannot be larger than $w_k$ by definition of $w_k$.

Suppose that instead $j < k$. Let $W'_j \subseteq W_j$ be the set of players in $W_j$ that win infinitely often. Then there must be $t'_0$ such that for every $t \geq t'_0$, no player in $\left(W_j \setminus W'_j\right) \cup \bigcup_{i=j+1}^k W_i$ is the winner at step $t$, and thus her stake never increases after this step. On the other hand, players in $W'_j$ wins infinitely often, and thus eventually we reach a time step $t_0 \geq t'_0$ such that for every $W \subseteq W_j$, we have that
\[\sum_{i \in W} \sigma_i \geq \theta_W \sum_{i \in W \cup \left(W_j \setminus W'_j \right) \cup \bigcup_{\ell = j+1}^k W_\ell} \sigma_i,\]
where $\theta_W$ is as defined in Condition~\ref{dm-subset} of Assumption~\ref{assmpt:dm}, and thus, according to this assumption, for every $\overline{W} \subseteq \left(W_j \setminus W'_j \right) \cup \bigcup_{\ell = j+1}^k W_\ell$ we have that $d(\tm{\sigma}_{W \cup \overline{W}}) = d(\tm{\sigma}_W)$. However, by definition of $W_j$ and the fact that a player in $W'_j$ wins at step $t$, we have that $\tm{P} = W \cup \overline{W}$ for some $W \subseteq W'_j$ and $\overline{W} \subseteq \left(W_j \setminus W'_j \right) \cup \bigcup_{\ell = j+1}^k W_\ell$\footnote{Actually, by Lemma~\ref{lem:equilibria}, we can claim that $\overline{W}$ must be exactly $\left(W_j \setminus W'_j \right) \cup \bigcup_{\ell = j+1}^k W_\ell$.}, and thus $d(\tm{\sigma_{\tm{P}}}) = d(\tm{\sigma}_W) \leq w_j$, where last inequality follows from Condition~\ref{dm-no-increase3} of Lemma~\ref{lem:tau_prop}.

The arguments above allow to claim that at each time step $t$ for $t$ sufficiently large, in which a player of $W_j$ wins the decentralization will be at most $w_j$. However, this does not exclude that decentralization may increase in the rounds in which a player of $W_1, \ldots, W_{j-1}$ wins. However, this is impossible by the properties of the decentralization function. Consider indeed, the sequence $t_1, \ldots, t_m$ of time steps in which a player in $W_i$ wins, for $i = 1, \ldots, j-1$, among two time steps in which a player in $W_\ell$ for $\ell > i$ wins. Note that at time step $t_1 - 1$, no player in $W_i$ participate, that in turn, by Lemma~\ref{lem:equilibria}, implies that for each player $x \in W_i$, $\sigma_x > \max_{y \in \TM{P}{t_1-1}} \sigma_w$ and the participation set $P^{\geq x}$ is harmful for $x$, and thus, by Lemma \ref{lem:harmful-d}, $d(\TM{\sigma_{P^{\geq x}}}{t_1-1}) < d(\TM{\sigma_{\TM{P}{t_1-1}}}{t_1-1})$, where $\TM{P}{t_1-1} \subseteq P^{\geq x}$. 

Then, by Condition~\ref{dm-no-increase2} it follows that $d(\TM{\sigma_{P^{\geq x}}}{t_1}) \leq d(\TM{\sigma_{\TM{P}{t_1-1}}}{t_1-1})$, and thus $v(d(\TM{\sigma_{P^{\geq x}}}{t_1})) \leq v(d(\TM{\sigma_{\TM{P}{t_1-1}}}{t_1-1})) \leq w_j$ for every $x \in W_i$. For the next steps $t_z$ in the sequence we can repeat the same argument if the winner at this step were not participating at step $t_{z-1}$ (since, by construction, among steps $t_{z-1}$ and $t_z$ no player with in $P^{\geq x}$ wins and thus changes her stake), and we can simply apply Condition~\ref{dm-no-increase1} otherwise.

Let us now consider the case that $\varepsilon > 0$. We say that an \emph{error} occurs at time step $t$ if a player in $W_j$ wins even if there are players from $W_i$, for $i < j$, who are participating. Note that, by Condition~\ref{dm-no-increase2}, an error may increase the decentralization of at most one. 
Note, moreover, that in case an error occurs, any agent participating at time $t$ still has an incentive to participate at time $t+1$. Not only, but each positive effect on decentralization due to an error at time step $t$ is canceled out after $\left\lceil\frac{\tau}{1 - \tau}\right\rceil$ time steps following $t$ in which an error does not occur. Indeed, let $t'$ be the last of these time steps. Note that $\sum_{i \in \tm{P}} \TM{\sigma}{t'} = \sum_{i \in \tm{P}} \tm{\sigma} + (t'-t+1)\beta$. Suppose that decentralization before the error at time step $t$ occurs is $d$. Then there must be a set $S$ of size $d$ such that $\sum_{i \in S} \tm{\sigma} > \tau \sum_{i \in \tm{P}} \tm{\sigma}$. Now observe that $\sum_{i \in S} \TM{\sigma}{t'} = \left\lceil\frac{\tau}{1 - \tau}\right\rceil \beta + \sum_{i \in S} \tm{\sigma}$, and thus $\sum_{i \in S} \TM{\sigma}{t'} > \tau \sum_{i \in \tm{P}} \tm{\sigma} + \left(t'-t-\left\lceil\frac{\tau}{1 - \tau}\right\rceil\right)\beta$, where we used that $\left\lceil\frac{\tau}{1 - \tau}\right\rceil \beta \geq \left(\left\lceil\frac{\tau}{1 - \tau}\right\rceil + 1\right)\tau \beta$.

Let $t_0$ be as described above for the case $\varepsilon = 0$. Then, for every $\varepsilon < \frac{1}{2}$, we can always choose $t^\varepsilon_0$ such that within $t^\varepsilon_0$ steps the ratio between the number of steps in which no error occurs and the number of steps in which an error occurs is at least $t_0 \cdot \left\lceil\frac{\tau}{1 - \tau}\right\rceil$  with overwhelming probability. Then, by the argument above, after $t^\varepsilon_0$ steps, the system value is at most $w$,  with overwhelming probability. Moreover, for every $\varepsilon < \frac{1}{2}$, there exists $q^\varepsilon$ such that within a sequence of $q^\varepsilon$ steps, the ratio between the number of steps in which no error occurs and the number of steps in which an error occurs is at least $\left\lceil\frac{\tau}{1 - \tau}\right\rceil$. Then, we can conclude that even if for some $t \geq t^\varepsilon_0$ the value of the system goes above $w$ it goes back to $w$ within at most $q^\varepsilon$ steps.
\end{proof}

\subsection{Proof of Theorem~\ref{thm:fails_myopic}}
\begin{proof}
Let us first consider the case of a deterministic policy ($\varepsilon=0$), that is,  the  winner is the one participant with largest type. 
Let $W$ denote the set of players who win infinitely many times and let $w_{\min}\in W$ be the one with the lowest type in this set. By definition of $W$, there exists a time $t_W$ after which (i) only players in $W$ win and (ii) the stake of the other players does not increase anymore. We observe that, for any $t>t_W$ in which $w_{\min}$ wins,  the following hold: 
\begin{enumerate}
    \item No other player in $W$ participates,  otherwise the latter would be the winner;  
    \item Hence, decentralization is measured according to the stake of $w_{\min}$ and only some of the players \emph{not} in $W$. 
\end{enumerate}
Since the stake of $w_{\min}$ keeps increasing, while that of players not in $W$ remains constant, we conclude that  eventually decentralization is at its minimum  in every  step in which $w_{\min}$ wins. 

Next we argue that the set $W$ can only contain the player with largest type. This implies that, after $t_W$, the only player to win is the player with largest type. By contradiction, suppose $W$ contains at least two players. Then, for every step in which $w_{\min}$ wins and the decentralization is at its minimum, it would have been better for the other players in $W$ to participate: Indeed, they would win and increase their own stake, while the decentralization cannot get worse because it is already at its minimum.  

% Finally, the theorem also holds for $\mu^\alpha$ with $\alpha$ sufficiently close to $1$, such that with probability close to $1$ the following two events hold: (i) the participating and winning players are as for the deterministic case above until some time $t_M\geq t_W$ in which decentralization stays at its minimum, and (ii) after $t_M$ the probability of a different winner is small enough to keep decentralization at minimum, even though players not in $W$ may occasionally win.
Consider now the case that $\varepsilon > 0$.
Let us consider $W$ and $w_{\min}$ exactly as above. We also set $\overline{W}=[n] \setminus W$.
We will show that we still have that there exist $t_0$ such that, with overwhelming probability, the decentralization is at its minimum in every time step in which $w_{\min}$ is the participant with highest type (even if she does not win).

To this end, let us consider the time steps $t_1, t_2, \ldots$ in which $w_{\min}$ is the participant with highest type. Let also denote with $X_i$ the fraction of stake that $w_{\min}$ has with respect to the stake of all participants with type worse that $w_{\min}$ at step $t_i$, i.e., $X_i = \frac{\TM{\sigma_{w_{\min}}}{t_i}}{\TM{\sigma_{w_{\min}}}{t_i} + \sum_{j \in \overline{W}}\TM{\sigma_j}{t_i}}$.

We first observe that the expected value of $X_i$ is increasing in $i$. 
Indeed, let $p_{i, \Delta, \delta}$ be the probability that between time step $t_i$ and $t_{i+1}$ there have been $\Delta$ steps in which the winner has been selected within $\overline{W} \cup \{w_{\min}\}$, and among them, in $\delta$ steps $w_{\min}$ has been selected as a winner. Then
{\allowdisplaybreaks
\begin{align*}
 & \mathbb{E}[X_{i+1}] = \\ & \sum_{\Delta = 0}^{t_{i+1}-t_i} \sum_{\delta = 0}^{\Delta} p_{i, \Delta, \delta} \left[(1-\varepsilon)\frac{\TM{\sigma_{w_{\min}}}{t_i}+(\delta+1) \beta}{\TM{\sigma_{w_{\min}}}{t_i} + (\Delta+1) \beta + \sum_{j \in \overline{W}} \TM{\sigma_j}{t_i}} + \varepsilon \frac{\TM{\sigma_{w_{\min}}}{t_i}+\delta \beta}{\TM{\sigma_{w_{\min}}}{t_i} + (\Delta+1) \beta + \sum_{j \in \overline{W}} \TM{\sigma_j}{t_i}}\right]\\
 & \geq (1-\varepsilon)\frac{\TM{\sigma_{w_{\min}}}{t_i}+\beta}{\TM{\sigma_{w_{\min}}}{t_i} + \beta + \sum_{j \in \overline{W}} \TM{\sigma_j}{t_i}} + \varepsilon \frac{\TM{\sigma_{w_{\min}}}{t_i}}{\TM{\sigma_{w_{\min}}}{t_i} + \beta + \sum_{j \in \overline{W}} \TM{\sigma_j}{t_i}} = \frac{\TM{\sigma_{w_{\min}}}{t_i}+(1-\varepsilon)\beta}{\TM{\sigma_{w_{\min}}}{t_i} + \beta + \sum_{j \in \overline{W}} \TM{\sigma_j}{t_i}}\\
 & = \left(1 - \frac{\beta}{\TM{\sigma_{w_{\min}}}{t_i} + \beta + \sum_{j \in \overline{W}} \TM{\sigma_j}{t_i}}\right)\frac{\TM{\sigma_{w_{\min}}}{t_i}\beta}{\TM{\sigma_{w_{\min}}}{t_i} + \sum_{j \in \overline{W}} \TM{\sigma_j}{t_i}} + \frac{(1- \varepsilon)\beta}{\TM{\sigma_{w_{\min}}}{t_i} + \beta + \sum_{j \in \overline{W}} \TM{\sigma_j}{t_i}}\\
 & = \left(1 - \frac{\beta}{\TM{\sigma_{w_{\min}}}{t_i} + \beta + \sum_{j \in \overline{W}} \TM{\sigma_j}{t_i}}\right)\mathbb{E}[X_{i}] + \frac{(1- \varepsilon)\beta}{\TM{\sigma_{w_{\min}}}{t_i} + \beta + \sum_{j \in \overline{W}} \TM{\sigma_j}{t_i}}
\end{align*}
}
\begin{comment}
By iterating the argument, and using that $\mathbb{E}[X_1] \geq 0$ and letting $S = \TM{\sigma_{w_{\min}}}{t_1} + \sum_{j \in \overline{W}} \TM{\sigma_j}{t_1}$, we achieve that
\[
 \mathbb{E}[X_{i+1}] \geq \sum_{j = 1}^{i+1} \frac{(1-\varepsilon)\beta}{S + j\beta} = (1-\varepsilon)\sum_{j = \frac{S}{\beta} + 1}^{\frac{S}{\beta} + i + 1} \frac{1}{j}.% = \Theta(\log i).
\]
Since the harmonic series is increasing and divergent, then we have that $\mathbb{E}[X_{i+1}]$ eventually exceeds any fixed constant. In particular, there must be $i_0$, such that for every $i \geq i_0$, $E[X_i] \geq \theta_{\{i\}} + \lambda$, where $\theta_{\{i\}}$ is defined by Condition~\ref{dm-subset} of Assumption~\ref{assmpt:dm}. 
\end{comment}
Write
\[
a_i \;:=\; \frac{\beta}{\TM{\sigma_{w_{\min}}}{t_i}+\beta+\sum_{j\in\overline{W}}\TM{\sigma_j}{t_i}}\ \in\ (0,1),
\]
so that the recursion above reads
$\mathbb{E}[X_{i+1}]=(1-a_i)\mathbb{E}[X_i]+(1-\varepsilon)a_i$. Subtracting
$1-\varepsilon$ from both sides gives
\[
\mathbb{E}[X_{i+1}]-(1-\varepsilon)=(1-a_i)\big(\mathbb{E}[X_i]-(1-\varepsilon)\big),
\]
and hence, iterating and using $\big|\mathbb{E}[X_1]-(1-\varepsilon)\big|\le1$,
\[
\big|\mathbb{E}[X_{i+1}]-(1-\varepsilon)\big|
\ \le\ \prod_{j\le i}(1-a_j)
\ \le\ \exp\Big(-\sum_{j\le i}a_j\Big).
\]
Since at most one token is awarded per round,
$\TM{\sigma_{w_{\min}}}{t_i}+\sum_{j\in\overline{W}}\TM{\sigma_j}{t_i}\le S+i\beta$,
so $a_i\ge\beta/(S+(i+1)\beta)$ and $\sum_i a_i$ diverges. Therefore
$\prod_j(1-a_j)\to0$ and $\mathbb{E}[X_i]\to1-\varepsilon$.

Let $\theta:=\theta_{\overline{W}\cup\{w_{\min}\}}<1$ be the threshold given by
Condition~2 of Assumption~3. Fix $\lambda>0$ with $\theta+\lambda<1$, and take
$\varepsilon$ small enough that $1-\varepsilon>\theta+\lambda$. Then there is
$i_0$ such that $\mathbb{E}[X_i]\ge\theta+\lambda$ for every $i\ge i_0$. 

We next prove that the actual value of $X_i$ is, in addition, concentrated around its expected value. To this aim, first observe that
{\allowdisplaybreaks
\begin{align*}
 & \mathbb{E}[X_{i+1}^2] =\\ & \sum_{\Delta = 0}^{t_{i+1}-t_i} \sum_{\delta = 0}^{\Delta} p_{i, \Delta, \delta} \left[(1-\varepsilon)\left(\frac{\TM{\sigma_{w_{\min}}}{t_i}+(\delta+1) \beta}{\TM{\sigma_{w_{\min}}}{t_i} + (\Delta+1) \beta + \sum_{j \in \overline{W}} \TM{\sigma_j}{t_i}}\right)^2 + \right. \\ & \left. \varepsilon \left(\frac{\TM{\sigma_{w_{\min}}}{t_i}+\delta \beta}{\TM{\sigma_{w_{\min}}}{t_i} + (\Delta+1) \beta + \sum_{j \in \overline{W}} \TM{\sigma_j}{t_i}}\right)^2\right]\\
 & \leq \left[(1-\varepsilon)\left(\frac{\TM{\sigma_{w_{\min}}}{t_i}+\beta}{\TM{\sigma_{w_{\min}}}{t_i} + \beta + \sum_{j \in \overline{W}} \TM{\sigma_j}{t_i}}\right)^2 + \varepsilon \left(\frac{\TM{\sigma_{w_{\min}}}{t_i}}{\TM{\sigma_{w_{\min}}}{t_i} + \beta + \sum_{j \in \overline{W}} \TM{\sigma_j}{t_i}}\right)^2\right] + \frac{\varepsilon}{1-\varepsilon}\\
 & = \frac{\left(\TM{\sigma_{w_{\min}}}{t_i}\right)^2+(1-\varepsilon)\beta^2 + 2(1-\varepsilon)\beta\TM{\sigma_{w_{\min}}}{t_i}}{\left(\TM{\sigma_{w_{\min}}}{t_i} + \beta + \sum_{j \in \overline{W}} \TM{\sigma_j}{t_i}\right)^2} + \frac{\varepsilon}{1-\varepsilon},
\end{align*}
}
where in the first inequality we used that for every $\Delta > 0$, it holds that $\sum_{\delta = 0}^\Delta p_{i, \Delta, \delta} = \varepsilon^\Delta (1-\varepsilon)^{t_{i+1}-t_i-\Delta} \leq \varepsilon^\Delta$, and thus $\sum_{\Delta=1}^{t_{i+1}-t_i} \sum_{\delta = 0}^\Delta p_{i, \Delta, \delta} \leq \sum_{\Delta=1}^{t_{i+1}-t_i} \varepsilon^\Delta = \frac{\varepsilon(1-\varepsilon^{t_{i+1}-t_i})}{1-\varepsilon} \leq \frac{\varepsilon}{1-\varepsilon}$. Hence, we achieve that
\begin{align*}
    & \mathrm{Var}[X_{i+1}] = \mathbb{E}[X_{i+1}^2] - \mathbb{E}[X_{i+1}]^2\\
    & \leq \frac{\left(\TM{\sigma_{w_{\min}}}{t_i}\right)^2+(1-\varepsilon)\beta^2 + 2(1-\varepsilon)\beta\TM{\sigma_{w_{\min}}}{t_i}}{\left(\TM{\sigma_{w_{\min}}}{t_i} + \beta + \sum_{j \in \overline{W}} \TM{\sigma_j}{t_i}\right)^2} + \frac{\varepsilon}{1-\varepsilon} - \frac{\left(\TM{\sigma_{w_{\min}}}{t_i}\right)^2+(1-\varepsilon)\beta^2 + 2(1-\varepsilon)\beta\TM{\sigma_{w_{\min}}}{t_i}}{\left(\TM{\sigma_{w_{\min}}}{t_i} + \beta + \sum_{j \in \overline{W}} \TM{\sigma_j}{t_i}\right)^2}\\
    & = \frac{\varepsilon(1-\varepsilon)\beta^2}{\left(\TM{\sigma_{w_{\min}}}{t_i} + \beta + \sum_{j \in \overline{W}} \TM{\sigma_j}{t_i}\right)^2} + \frac{\varepsilon}{1-\varepsilon} \leq \frac{2\varepsilon}{1-\varepsilon}.
\end{align*}
Then, by Chebishev Inequality we achieve that for every $i \geq i_0$,
\[
 \mathbb{P}(X_i < \theta_{\{i\}}) \leq \mathbb{P}(|X_i - E[X_i]| > \lambda) \leq \frac{\mathrm{Var}[X_{i}]}{\lambda^2} \leq \frac{2\varepsilon}{\lambda^2(1-\varepsilon)}.
\]
Since $\lambda$ is a fixed constant and the bound $2\varepsilon/\lambda^2$ tends to $0$ with $\varepsilon$,
%Since $\varepsilon$ is negligible and $\lambda$ is constant, 
we can conclude that, with overwhelming probability, %, except with negligible probability, 
in the time step $t_i$, with $i \geq i_0$, in which $w_{\min}$ is the participant with highest type, it holds that $w_{\min}$ has a fraction at least $\theta_{\{i\}}$ of the participants' total stake, and thus, by Condition~\ref{dm-subset} of Assumption~\ref{assmpt:dm}, we have that in all these steps the decentralization is $d(\TM{\sigma_{w_{\min}}}{t_i})$, that, by Condition~\ref{dm-singleton} of Assumption~\ref{assmpt:dm}, is at its minimum.

The claim then follows by observing, exactly as above, that $w_{\min}$ must necessarily be the player with highest type, and hence, the decentralization %system value 
is at its minimum for every step, %except 
with %neglible 
overwhelming probability, as desired.
\end{proof}

% \paragraph{Proof of Theorem~\ref{thm:fails_myopic}. - EXTENSION 2 (top-k win)} ToDo

\subsection{Proof of Theorem \ref{thm:fails_topk}}
\begin{proof}
Let $W$ denote the set of players who are among the top-$k$ participants infinitely many times. Note that $|W| \geq k$.
% and let $w_{\min}\in W$ be the one with the lowest type in this set. By definition of $W$, there exists a time $t_W$ after which (i) only players in $W$ are among the top-$k$ participants and (ii) the stake of the other players does not increase anymore. We observe that, for any $t>t_W$ in which $w_{\min}$ is among the top-$k$ participants,  the following holds: 
% \begin{enumerate}
%     \item At most $k-1$ players in $W$ different from $w_{\min}$ participate;  
%     \item Hence, decentralization is measured according to the stake of $w_{\min}$, the stake of at most $k-1$ players in $W$ participating, and the stake of the participating players \emph{not} in $W$. 
% \end{enumerate}
% Since the stake of players in $W$ keeps increasing, while the stake of players not in $W$ remains constant, we conclude that there is a time $t^*_W \geq t_W$ such that for every $t \geq t^*_W$ in which $w_{\min}$ is among the top-$k$ participants, the system value is at most $\vartheta$. 
Let $K$ be the set of $k$ players with largest type.
We first prove that $K \subseteq W$.
Suppose, by sake of contradiction, that this is not the case and there is $j \in K \setminus W$. Let $t_j$ be the last time step in which $j$ is among the top $k$ participants.
Since $|W| \geq k$, and $j \notin W$, there must exist $i \in W \setminus K$.
Since $i$ is among the top $k$ participants infinitely often, and $j$ is among the top $k$ participants only a finite number of times, there will certainly be a time step $t > t_j$ such that $\tm{\sigma_i} > \tm{\sigma_j}$. However, by Lemma~\ref{lem:equilibria}, at time step $t$, if player $i$ participates, then also $j$ must participate, and hence it must be among the top $k$ participants. That contradicts our definition of $t_j$, and hence proves that $K \subseteq W$, as desired.

Let us now focus only on the infinitely many time steps $t_1, t_2, \ldots$ in which players in $K$ are exactly the top-$k$ participants. 
We next prove that for every $i$ such that decentralization at $t_i$ is below $\vartheta$, there exists with overwhelming probability $i^* > i$ such that the decentralization will be at least $\vartheta$ for every $t_j$ such that $j \geq i^*$. This proves not only that the policy will eventually lead to have decentralization $\vartheta$, but also that it will keep this value with overwhelming probability in all of the following steps, and even in the case that this does not occur (i.e., the value happens to go below $\vartheta$), the policy will eventually recover value $\vartheta$. This allows us to conclude that also the second part of the theorem is true, namely that the with overwhelming probability the decentralization will be at least $\vartheta$ for the all the time step $t_1, t_2, \ldots$, except a few ones, namely the ones between $t_i$ and $t_{i^*}$. 

To prove that $i^*$ exists, note that in expectation each of these players should have gained the same amount of tokens within these steps. Hence, there must be $t^\circ$ such that for every $t_i \geq t^\circ$ the probability that the amount of tokens gained by two different players in $K$ differs by more than a small constant $\delta$ is negligibly small. 
Assume without loss of generality that $\beta=1$, and let $m$ denote the number
of rounds elapsed since $t^{\circ}$. Then, with overwhelming probability, the
stake of each player $j\in K$ lies in
$\big[\sigma_j^{0}+\tfrac{m}{k}-\delta,\ \sigma_j^{0}+\tfrac{m}{k}+\delta\big]$,
where $\sigma_j^{0}$ is the stake of $j$ at $t^{\circ}$. Writing
$\Delta:=\max_{j,j'\in K}\big|\sigma_j^{0}-\sigma_{j'}^{0}\big|+2\delta$ for the
spread across $K$, the stakes of the players in $K$ therefore agree up to
$\Delta$, while their common value grows as $m/k$. Since $\Delta$ is fixed,
$\Delta k/m\to0$, so there is $i^{*}$ such that for every $i\ge i^{*}$
the players in $K$ hold equal stakes up to a vanishing fraction of their value.
Hence, for every $t_i>t_{i^{*}}$, with overwhelming probability the
decentralization is at least $\vartheta$.

%Then there must be $i^*$, such that $\sigma^0_j \pm \delta$ is negligibly small with respect to $\frac{i^*}{k}$. Hence, for every $t_i > t_{i^*}$, it occurs that, with overwhelming probability, the decentralization would be at least $\vartheta$.

In order to conclude the proof, we need to show the decentralization is at least $\vartheta$ even in all steps $t \in [t_i, t_{i+1}]$ such that the decentralization at $t_i$ is at least $\vartheta$. Indeed, in each of these $t$ the set of top-$k$ participants does not coincide with $K$, and thus there must be $j \notin K$ that does not participate. However, by Lemmas~\ref{lem:harmful-d} and ~\ref{lem:equilibria}, in order for this to happen, it must be that the decentralization when she does not participate must be at least $\vartheta$. %\footnote{Maybe this claim only holds for tau decentralization index, or it requires some assumption on decentralization, namely the delta of dencentralization measure is constant (for tau-decentralizaiton indices is 1)}, as desired.
\end{proof}

\subsection{Proof of Theorem \ref{thm:fails_topkappa}}
\begin{proof}
Let $W$ and $K$ be as in the proof of Theorem~\ref{thm:fails_topk}, and observe that, as proved therein, $K \subseteq W$.

% Let $W$ denote the set of players who are among the top-$\kappa$ participants infinitely many times and let $w_{\min}\in W$ be the one with the lowest type in this set. By definition of $W$, there exists a time $t_W$ after which (i) only players in $W$ are among the top-$\kappa$ participants and (ii) the stake of the other players does not increase anymore. We observe that, for any $t>t_W$ in which $w_{\min}$ is among the top-$\kappa$ participants,  the following holds: 
% \begin{enumerate}
%     \item At most $\kappa-1$ players in $W$ different from $w_{\min}$ participate;  
%     \item Hence, decentralization is measured according to the stake of $w_{\min}$, the stake of at most $\kappa-1$ players in $W$ participating, and the stake of the participating players \emph{not} in $W$. 
% \end{enumerate}
% Since the stake of players in $W$ keeps increasing, while that of players not in $W$ remains constant, we conclude that there is a time $t^*_W \geq t_W$ such that for every $t \geq t^*_W$ in which $w_{\min}$ is among the top-$\kappa$ participants, the system value is at most $\vartheta$. We next prove that there would be a time step $t^\circ_W$ such that the system value is at least $\vartheta$ and it will never goes below this threshold.

Let us consider $w_{\max}$, the one player in $K$ with largest initial stake.
Since players in $K$ will be among the top-$k$ participants infinitely often, and the policy always chooses as winner the one among the top-$k$ participants with the lowest stake, then each player in $K$ wins infinitely often. In particular, this must happen for $w_{\max}$. However, in order for this to happen, by our policy, it must be that there is a time step $t^\circ$ such that $\TM{\sigma_{w_{\max}}}{t^\circ} \leq \TM{\sigma_j}{t^\circ} \leq \TM{\sigma_{w_{\max}}}{t^\circ}+1$ for every $j \in K$. Then the decentralization %value of the system 
at $t^\circ$ will be at least $\vartheta$ if all players in $K$ will participate.

Hence, by Lemma \ref{lem:harmful-d}, participation cannot be harmful for any player in $K$, and thus we have the set of participants at time $t^\circ$ will contain $K$. Moreover, by our policy, the stake among these agents will remain balanced for any successive step, implying that at each of these steps both the set of participants still contain all players in $K$, and the decentralization %value of the system 
will be at least $\vartheta$.
\end{proof}

\subsection{Proof of Theorem \ref{thm:broadsharing}}
We first state a useful equal-increment monotonicity property of $\tau$-decentralization indices.

\begin{lemma}\label{lem:equalincrements}
Let decentralization be measured through a $\tau$-decentralization index and let
$c>0$. Then $d(\tm{\sigma}+c\mathbf{1})\ge d(\tm{\sigma})$, where $\mathbf{1}$ is
the all-ones vector.
\end{lemma}

\begin{proof}
Adding the same amount to every stake preserves the ordering of players, so assume
$\tm{\sigma_1}\ge\cdots\ge\tm{\sigma_n}$ and write $A_k=\sum_{i\le k}\tm{\sigma_i}$
and $\tm{S}=A_n$. For every $k$,
\begin{align*}
\frac{A_k+kc}{\tm{S}+nc}\ & \le\ \frac{A_k}{\tm{S}} 
\iff \tm{S}\,(A_k+kc)\le A_k(\tm{S}+nc) \\
& \iff \tm{S}k\le A_k n  \iff \frac{k}{n}\le\frac{A_k}{\tm{S}}, 
\end{align*}
and the last inequality holds because the $k$ largest stakes have average at least
the average of all $n$. Hence every top-$k$ share weakly decreases, so the set
$\{k:A_k/\tm{S}>\tau\}$ can only shrink and its minimum can only grow, i.e., 
$d(\tm{\sigma}+c\mathbf{1})\ge d(\tm{\sigma})$.
\end{proof}

\begin{proof}[Proof of Theorem \ref{thm:broadsharing}]
By Condition~3 of Assumption~\ref{assmpt:dm},
$d(\tm{\sigma})\ge d(\tm{\sigma}_{-i^*})$ implies $d(\tm{\sigma})\ge
d(\tm{\sigma}_{-i})$ for every player $i$. Since $\mu$ is broad-sharing,
$\tm{B_i}>0$, and since $v$ is non-decreasing,
\[
\big(\tm{\sigma_i}+\tm{B_i}\big)\,v\big(d(\tm{\sigma})\big)
\ \ge\ \tm{\sigma_i}\,v\big(d(\tm{\sigma}_{-i})\big),
\]
so no participant prefers to abstain, and full participation is an equilibrium.
Under $\mu^{\mathrm{all}}$ the winner is the largest-type participant by
definition, and under full participation every stake grows by $\beta/n$, so
$d(\TM{\sigma}{t+1})\ge d(\tm{\sigma})$ by Lemma~\ref{lem:equalincrements}.
\end{proof}

\section{Simulating 
% Players with 
Asymmetric Lookahead}\label{app:lookahead}

Let us consider asymmetric lookahead, a stylized behavioral model meant to capture a bias toward continued participation. The idea is that participation decisions may be taken using a short-run criterion, while abstention is treated as a more consequential deviation and is therefore evaluated using a forward-looking comparison of continuation paths.  
Formally, agents are myopic when deciding whether to participate, but when considering abstention they compare the current path with the equilibrium
% continuation
path that would follow from dropping out. 
\begin{definition}[Players with asymmetric lookahead]\label{def:alookahead}
	A player $i$ has \emph{asymmetric lookahead} if she 
	is myopic for participation, whereas for abstention she considers the equilibrium path of play. Specifically, for each time step $\ell \geq t+1$, the agent considers the\footnote{We will see below that this is unique and thus there is no ambiguity here.} equilibrium participation set $P^\ell$ for the single shot game at round $\ell$ until we reach a time step $\ell$ such that $i \in P^\ell$. Note that this sequence depends on the set $P$ of participants at time step $t$. We set $R_i(P) = (P_1, \ldots, P_k)$ to be the equilibrium sequence of participation sets discussed above, if this sequence eventually stops, and set $R_i(P) = \emptyset$ otherwise. If the sequence exists, i.e., $R_i(P) \neq \emptyset$, we
    % will 
    say that $i$ has a \emph{recovery plan} $R_i(P)$ from $P$ and let $k$ to
    be
    % denote 
    her length.
	
	Specifically, we say that for players with asymmetric lookahead:
	\begin{equation*}
		\tm{U_i}(P) =
		\begin{cases}
			\left(\tm{\sigma}_i + \tm{B_i}(P)\right)\cdot v(d(\tm{\sigma}_P)) %- \tm{c}_i 
            & \text{if $i \in \tm{P}$},\\
			-\infty & \text{if $i\not\in \tm{P}$ and $R_i(P) = \emptyset$},\\ 
			\tm{\sigma}_i \cdot \E\left[ v(d(\TM{\sigma}{t+k}_{P_{k}}))\right] & \text{otherwise}.
		\end{cases}
	\end{equation*}
	Then a set $\tm{P}$ of participants at round $t$ is an equilibrium for players with asymmetric lookahead if for every $i \in \tm{P}$, it holds that either $\tm{U_i}(\tm{P}) \geq \tm{U_i}(\tm{P} \setminus \{i\})$ or $i$ does not have a recovery plan in $\tm{P} \setminus \{i\}$, and for every $i \notin \tm{P}$, it holds that $\tm{U_i}(\tm{P}) < \tm{U_i}(\tm{P} \cup \{i\})$ and $i$ has a recovery plan.
\end{definition}
While the asymmetry in this behavioral model can be loosely motivated by prospect-theoretic status-quo effects \cite{kahneman1979prospect} and by participation-oriented norms in blockchain communities \cite{primavera}, and hence this behavioral model can have some interest in itself, in this work we use it only as a tool for proving that a policy exists guaranteeing full participation of myopic players at every step, that decentralization never falls below a given threshold, and that performance loss is bounded. Specifically, the policy achieving these properties essentially simulates what would have happened if players have asymmetric lookahead. For this reason, before to describe the policy of interest and to formalize its properties, it would be useful to analyze how the system evolves when players have asymmetric lookahed.

\noindent \emph{Equilibria characterization for players with asymmetric lookahead.}
% We next characterize the equilibrium for players with asymmetric lookahead. 

\begin{lemma}
\label{lem:equilibria:lookahead}
For every initial stake profile and every monetary policy $\mu$, if players have asymmetric lookahead then for every time step $t$ the single-shot game at round $t$ admits a
unique equilibrium, and this equilibrium is a suffix of players with respect to
the stake order. 
\end{lemma}
The lemma shows how to define recovery plans for abstaining players in a way that the properties of equilibria for myopic players are preserved for these more sophisticated agents.
% {\color{blue} However, these recovery plans come at a cost, as we cannot guarantee that equilibria can be computed in polynomial time in this case.}
\begin{proof}
Consider the notion of harmfulness for the utility of players with asymmetric lookahead (Definition \ref{def:alookahead}). Let us consider the following procedure, that returns the set $P$ of participants for the current time step: For each player $i$ in non-decreasing order of stake, if $\tm{\sigma}$ is not harmful for $i$ in $P^{\geq i}$, then set $P = P^{\geq i}$. Note that the procedure returns a suffix set of participants (recall that participants are in order of stake, and thus a suffix set of participants includes all players whose stake is below a given threshold).

We will prove that for every player $i \in P$ returned by the procedure above, $\tm{\sigma}$ is not harmful in $P$ (even if %in the procedure above 
harmfulness has been evaluated only against a smaller participation set). In turn, this means that no player $i \in P$ has any incentive to deviate. Consider instead $i \notin P$. We will prove that $i$ has a recovery plan. Hence, for $i$ it is beneficial to abstain. Hence, we can conclude that $P$ is in equilibrium.

Let us first prove that for every player $i \in P$, $\tm{\sigma}$ is not harmful in $P$.
To this aim let $j^*$ be the player in $P$ with larger stake.
Note that for $j^*$ participation is not harmful against $P$ by construction, since $j^*$ is the last player for whom $\tm{\sigma}$ has been evaluated to not be harmful (hence, $P = P^{\geq j^*}$ and the harmfulness of $i$ has been evaluated exactly against $P$). Hence we have that
$$
 v(d(\sigma_P)) \left(\tm{\sigma_{j^*}} + \tm{B_{j^*}}(P) %\sum_{\ell \in P} \tm{w_\ell}(P) \tm{b}_{j^*}(\ell
 \right) \geq v(d(\sigma_{P \setminus \{j^*\}})) \tm{\sigma_{j^*}}. 
$$ 
By Assumption~\ref{ass:aligned}, this implies that $v(d(\sigma_P)) \geq v(d(\sigma_{P \setminus \{j^*\}}))$, and thus, in turn, that $d(\sigma_P) \geq d(\sigma_{P \setminus \{j^*\}})$; hence, by the properties of decentralization measure (i.e., Condition \ref{dm-mon-remove} of Assumption \ref{assmpt:dm}) $d(\sigma_{P}) \geq d(\sigma_{P \setminus \{j\}})$ for all $j>j^*$, and thus $$v(d(\sigma_{P})) \left(\tm{\sigma_j} + %\sum_{x \in P} \tm{w_x}(P) \tm{b_j}(x)
\tm{B_j}(P)\right) \geq v(d(\sigma_{P\setminus\{j\}})) \tm{\sigma_j},$$
meaning that no such $j>j^*$ has an incentive to abstain and drop out of $P$.
 
Let us finish by providing the recovery plan for each non-participant player.
Consider the following procedure:
\begin{enumerate}%[nosep]
    \item Initialize $x = 1$;
    \item let $\TM{\sigma_j}{t+x} = \TM{\sigma_j}{t+x-1} + \TM{B_j}{{t+x-1}}(\TM{P}{t+x-1})$ %E_{b \in \mu_b(\TM{P}{t+x-1})}$ 
    for every $j$ \label{item:recovery};
    \item Run the procedure described above to compute the equilibrium set of participants with $\TM{\sigma}{t+x}$ in place of $\tm{\sigma}$, and let $\TM{P}{t+x}$ be its output;
    \item Increase $x$ and repeat from Line~\ref{item:recovery}.
\end{enumerate}
Now for a player $i \notin P$, we stop the procedure above after $X_i$ steps, where $X_i$ is the smallest value of $x$ for which $i \in \TM{P}{t+x}$. It is immediate to check, that as long as stakes are bounded, in a bounded number of steps we have to reach a stake profile that is not harmful for $i$. Indeed, after remaining participants have been selected a sufficiently large, but still bounded, number of times, they accumulated a stake so large that the participation of $i$ cannot decrease decentralization, and hence participation becomes no longer  harmful. 
\end{proof}

%\paragraph*
\smallskip \noindent \emph{On the policy $\mu^*$ against players with asymmetric lookahead.}
Interestingly, players able to look ahead can fix the shortcoming highlighted in Theorem~\ref{thm:fails_myopic}.
To this aim, we say that given an initial stake profile $\TM{\sigma}{1}$, and a player $i$, the \emph{threshold} $\theta_i$ of $i$ is the largest system value such that any system value below this would
% be so low to 
make the player unhappy to win, i.e., $\theta_i = \min_t \{v(d(\tm{\sigma})) \colon \tm{\sigma} \text{ is harmful for $i$}\}$
The threshold $\theta$ of an instance is $\min_i \theta_i$.~\footnote{Contrarily to $\vartheta$, $\theta$ is defined in terms of system value and not decentralization, as it depends on payoffs.}.
\begin{theorem}\label{thm:asymlookahead}
For every initial stake profile $\TM{\sigma}{1}$ the monetary policy $\mu^*$ eventually keeps the system value above the threshold $\theta$ if players have asymmetric lookahead.
\end{theorem}
%\begin{proof}
%\textcolor{blue}{ToDiscuss}
 We omit a proof of this theorem since it follows the proof of Theorem~\ref{thm:fails_myopic}. That argument works by assuming that players with the highest types will participate whenever by not participating the system would be at its minimum value.
 However, this assumption does not hold if the player has a recovery plan. Indeed, in this case the player would prefer to  not participate and to temporarily reduce the system value (even if this implies to have this value going to the minimum) as long as this would allow to recover a larger system value. Moreover, the recovery plan is actually implemented, since it is an equilibrium\footnote{Note that for probabilistic monetary policy, the recovery plan may not be actually implemented. Still the first step of the recovery plan will be implemented. The recovery plan can then be updated according to the realization of policy at the current round, and this updated recovery plan can be played in the next rounds.}. Hence the desired system value is eventually restored above the threshold needed for that player to not be harmful.  
%\end{proof}

Unfortunately, lookahead does not avoid the system  sometimes going below $\theta$, and thus the policy is unable to always keep the system value above $\theta$. Actually, we prove an even stronger result:
% , namely that 
with $\mu^*$ we eventually reach the minimum system value.
\begin{theorem}\label{thm:asymlookahead2}
For every initial stake profile $\TM{\sigma}{1}$, there exists a $t$ such that, if monetary policy is $\mu^*$, then the  %system value 
decentralization at time step $t$ is at its minimum value, even if players have asymmetric lookahead.
\end{theorem}
\begin{proof}
Let us first consider the case that $\varepsilon = 0$. As in the proof of Theorem~\ref{thm:fails_myopic}, let us consider the player $w_{\min}$ of smaller type among the ones winning infinitely often. Clearly if $w_{\min}$ is the one player with the smallest type, then this implies that it should be the only participating player, and hence, the decentralization is at its minimum value. On the other side if $w_{\min}$ is not the player with the smallest type, we have that one player, $w_{\min}$ becomes richer and richer, and all players with smaller type never increase their stake. Eventually, the stake of $w_{\min}$ cover almost entirely the total stake distributed among $w_{\min}$ and the player with smaller stake. Since, these are the only agents that can participate in order of $w_{\min}$ to win, we can conclude that eventually in all steps in which $w_{\min}$ wins, the decentralization must be at its minimum value.

As in the proof of Theorem~\ref{thm:fails_myopic}, the argument can be then extended to work even with $\varepsilon > 0$, with overwhelming probability.    
\end{proof}

As above we provide an example that highlights this behavior.
\begin{example}[Lookahead Participation with Nakamoto Index]\label{ex:lookahead participation}
Consider
% exactly
the same setting as in Example~\ref{ex:myopic participation}, except that now players have asymmetric lookahead.
% It is not hard to check 
Note
that for the first four rounds these players will take the same decisions as in the case of myopic players. However, for round 5 Players~1 and 2 can plan to let Player~3 win (at the cost to have the Nakamoto index reach its minimum value)  so that at next step
% the Nakamoto index
decentralization
would increase again, guaranteeing a larger utility to both of them. 
The dynamics is summarized in Table~\ref{tab:lookahead}.
\begin{table}[t]
    \centering
    \rowcolors{2}{gray!10}{white} 
    %\resizebox{\columnwidth}{!}{
\begin{tabular}{@{}c c c c c c@{}}
\toprule
\textbf{Round} & \textbf{Initial Stake} & \textbf{Participants} & \textbf{Nakamoto Index} & \textbf{Winner} & \textbf{Final Stake} \\
\midrule
1 & (1, 1, 1) & All three & 2 & Player 1 & (2, 1, 1) \\
2 & (2, 1, 1) & All three & 2 & Player 1 & (3, 1, 1) \\
3 & (3, 1, 1) & P2, P3     & 2 & Player 2 & (3, 2, 1) \\
4 & (3, 2, 1) & All three & 2 & Player 1 & (4, 2, 1) \\
5 & (4, 2, 1) & P3 & 1 & Player 3 & (4, 2, 2) \\
6 & (4, 2, 2) & All three & 2 & Player 1 & (5, 2, 2) \\
7 & (5, 2, 2) & P2, P3 & 2 & Player 2 & (5, 3, 2) \\
8 & (5, 3, 2) & All three & 2 & Player 1 & (6, 3, 2) \\
9 & (6, 3, 2) & P3 & 1 & Player 3 & (6, 3, 3) \\
10 & (6, 3, 3) & All three & 2 & Player 1 & (7, 3, 3) \\
\bottomrule
\end{tabular}
%}
    \caption{Evolution of stakes, participation, Nakamoto index, reward allocation with asymmetric lookahead}
    \label{tab:lookahead}
\end{table}
\end{example}

\noindent \emph{The simulating policy.}
We are now ready to define the lookahead-simulating benchmark policy $\mu^\ell$, which works as follows: for each round $t$, choose as the winner who would have been chosen by $\mu^*$ at round $t+1$ if players have asymmetric lookahead.

The distinction between myopic behavior and asymmetric lookahead is central to the construction of $\mu^\ell$. Under myopic play, agents evaluate participation only through its immediate
effect on rewards and current token value, which makes strongly performance-favoring
policies particularly prone to concentration. Asymmetric lookahead changes this logic:
agents still use a short-run criterion when deciding to participate, but they evaluate
abstention through the continuation path that follows from dropping out. This asymmetry is
what allows temporary abstention to sustain future recovery in system value, and it is the
reason why the benchmark policy $\mu^\ell$ can simulate the beneficial effect
of lookahead even when the underlying agents are myopic. For example, in the setting of Example~\ref{ex:lookahead participation}, this policy would work as described in Table~\ref{tab:simulation}.

\begin{table}[t]
    \centering
    \rowcolors{2}{gray!10}{white} 
    %\resizebox{\columnwidth}{!}{
\begin{tabular}{@{}c c c c c c@{}}
\toprule
\textbf{Round} & \textbf{Initial Stake} & \textbf{Participants} & \textbf{Nakamoto Index} & \textbf{Winner} & \textbf{Final Stake} \\
\midrule
1 & (1, 1, 1) & All three & 2 & Player 1 & (2, 1, 1) \\
2 & (2, 1, 1) & All three & 2 & Player 2 & (2, 2, 1) \\
3 & (2, 2, 1) & All three & 2 & Player 1 & (3, 2, 1) \\
4 & (3, 2, 1) & All three & 2 & Player 3 & (3, 2, 2) \\
5 & (3, 2, 2) & All three & 2 & Player 1 & (4, 2, 2) \\
6 & (4, 2, 2) & All three & 2 & Player 2 & (4, 3, 2) \\
7 & (4, 3, 2) & All three & 2 & Player 1 & (5, 3, 2) \\
8 & (5, 3, 2) & All three & 2 & Player 3 & (5, 3, 3) \\
9 & (5, 3, 3) & All three & 2 & Player 1 & (6, 3, 3) \\
10 & (6, 3, 3) & All three & 2 & Player 2 & (6, 4, 3) \\
\bottomrule
\end{tabular}
%}
    \caption{Evolution of stakes, participation, Nakamoto index, and reward allocation}
    \label{tab:simulation}
\end{table}

In this example, it happens that at all steps all three players are participating and the system value never decreases. Interestingly, this is not merely an artifact of the example. Indeed, this monetary policy achieves the desired guarantees when decentralization is measured via a $\tau$-Decentralization Index.  To prove this, we first state the following lemma.

\begin{lemma}\label{lem:good_recovery}
Let decentralization be measured via a $\tau$\nobreakdash-\hspace{0pt}decentralization index, and consider the recovery plan
$R_i=(P^{(t+1)},\dots,P^{(t+r)})$ for some player $i$. Then along this
recovery path the system value is weakly non-decreasing from one step to the next, namely
$
d(\sigma^{(x)}) \ge d(\sigma^{(x-1)})$
for every  $x\in\{t+1,\dots,t+r\}.
$
\end{lemma}
\begin{proof}
Fix $x \in \{t+1,\dots,t+r\}$ and let $d(\TM{\sigma}{x-1}) = d$. Let $S_d$ to denote the set of $d$ players with largest stake in $\TM{\sigma}{x-1}$. Since $d(\TM{\sigma}{x-1}) = d$, then the stake of players in $S_{d-1}$ is at most a fraction $\tau$ of the total stake. Let $w$ be the player selected as winner at time $x$. If $w \notin S_{d-1}$, then the stake of players in $S_{d-1}$ is at most a fraction $\tau$ of the new total stake (that is larger than previous total stake). Hence, $d(\TM{\sigma}{x}) \geq d$, as desired.

Consider instead that $w \in S_{d-1}$. Let $X = \sum_{j \in S_{d-1}} \TM{\sigma}{x-1}(j)$ and $Y = \sum_{j \notin S_{d-1}} \TM{\sigma}{x-1}(j)$. As stated above $X \leq \tau (X+Y$), or equivalently, $(1-\tau)X \leq \tau Y$. On the other side, by denoting with $j_d$ the player in $S_d \setminus S_{d-1}$, we must have that $X + \TM{\sigma}{x-1}(j_d) > \tau (X+Y)$, from which it follows that $\TM{\sigma}{x-1}(j_d) > \tau Y - (1-\tau)X$.

Suppose, by sake of contradiction, that $d(\TM{\sigma}{x}) < d(\TM{\sigma}{x-1})$. Then $X+\beta > \tau(X+Y+\beta)$, from which it follows that $\beta > \frac{\tau Y - (1-\tau)X}{1- \tau}$. In other words, $\beta$ is sufficiently large that we do not need to consider the contribution of $j_d$.

Let us now consider $d(\TM{\sigma_{\TM{P}{x}}}{x-1})$ and $d(\TM{\sigma_{\TM{P}{x}}}{x})$. In the former case it would be impossible to have a fraction larger than $\tau$ of the total stake without player $j_d$, while for the latter the stake of $w$ is large enough to guarantee that the desired property is guaranteed even without  the contribution of $j_d$. Hence, we have that $d(\TM{\sigma_{\TM{P}{x}}}{x}) < d(\TM{\sigma_{\TM{P}{x}}}{x-1})$. Hence, by Assumption~\ref{ass:aligned}, participation of $w$ is causing a loss to $w$. Since, as showed in the proof of Theorem~\ref{thm:asymlookahead}, $w$ has a recovery plan guaranteeing the value for $w$ to be recovered, then it is convenient for $w$ to not participate, that contradicts the fact that it has been selected as the winner.
\end{proof}

% \begin{proof}[Proof sketch]
% Fix a step $x$ of the recovery path, and let $d=d(\sigma^{(x-1)})$. Consider the winner
% selected at that step.
% % 
% If the winner lies outside the current top $d-1$ stakeholders, then increasing her stake
% cannot reduce the $\tau$-decentralization index, because the coalition responsible for
% crossing the $\tau$ threshold is unchanged. 

% %Hence $d(\sigma^{(x)})\ge d(\sigma^{(x-1)})$.

% Suppose instead that the winner lies inside the current top $d-1$ stakeholders, and assume
% for contradiction that $d(\sigma^{(x)})<d(\sigma^{(x-1)})$. Then awarding the reward to that
% winner makes concentration strictly worse. In particular, this means that the winner is
% participating in a harmful configuration: by abstaining, she would avoid the loss in system
% value and rely on the continuation path that restores decentralization. But this contradicts
% the definition of a recovery plan, under which the prescribed participants are precisely those
% for whom participation is consistent with the continuation equilibrium. 
% %
% Therefore the system value cannot decrease at any step of the recovery path.
% \end{proof}

Lemma~\ref{lem:good_recovery} is what allows the simulated-lookahead policy to preserve the threshold in the next theorem.

%The key invariant behind $\mu^\ell$ is that, along the recovery path induced by asymmetric lookahead, the $\tau$-decentralization index does not decrease from one step to the next. Intuitively, if the selected winner lies outside the current top $d-1$ stakeholders, then increasing her stake cannot reduce the index. If instead the selected winner lies inside the top $d-1$ and this were to reduce decentralization, then that winner would be participating in a harmful configuration; but along a recovery path such a player would prefer to abstain, contradicting the definition of the recovery plan. This monotonicity is formalized in the appendix and is what allows the simulated-lookahead policy to preserve the threshold in the next theorem.

%as long as the token value function satisfies the following \emph{smoothness} property:

\begin{proof}[Proof of Theorem~\ref{thm:asymlookahead3}]
Let us denote with $\tm{\sigma'}$ the stake profile at round $t$ if policy $\mu^*$ was adopted.
 Recall that the winner selected at round $t$ by $\mu^\ell$ is exactly the same selected by $\mu^*$ at round $t+1$.  Let us denote with $w_t$ the winner at time $t$ according to $\mu^\ell$, and with $I^\beta_{w_t}$ the vector such that $I^\beta_w(w_t) = \beta$ and $I^\beta_w(i) = 0$ for every $i \neq w_t$.
 Then $\TM{\sigma}{t+1} = \tm{\sigma} + I^\beta_{w_{t}}$ and $\TM{\sigma'}{t+2} = \TM{\sigma'}{t+1} + I^\beta_{w_{t}}$.
 Let $d_1 = d(\TM{\sigma}{1})$. We prove that for each round $t \geq 1$, $v(d(\tm{\sigma})) \geq v(d_1) \geq \theta$. Indeed, if $w_{t}$ is the player with largest stake, the claim follows since, by Assumption~\ref{ass:aligned}, $v(d(\tm{\sigma})) \geq v(d(\TM{\sigma}{t-1}))$, and, by inductive hypothesis, $v(d(\TM{\sigma}{t-1})) \geq v(d_1)$. If instead $w_t$ is not the player with largest stake, then let $t' < t+1$ be the last round for which participation of the agent of larger stake is not harmful in $\mu^*$. Note, that by our assumption on the starting profile, $t'$ is well-defined. Then $w_t$ is the winner selected in the $(t+1-t')$-th step of the recovery plan of $t'$. By Lemma~\ref{lem:good_recovery}, we then have that $v(d(\tm{\sigma})) \geq v(d(\TM{\sigma}{t-1}))$. By merging this with the inductive hypothesis the claim follows.
 
 Moreover, since in $\mu^\ell$ the winner does not depend on who is participating, then participating is always a non-dominated strategy, even for myopic players.
 We conclude the proof by observing that, since we are simulating $\mu^*$ we have exactly the same winners of this policy, simply shifted by one round. Hence, the performances of $\mu^\ell$ tend to the $\mu^*$ as the number of rounds becomes large enough to discount the first non-simulated round.
 \end{proof}

\subsection{On sybilproofness of $\mu^{\ell}$} \label{sec:muell:sybilpf}
%Proof of Theorem~\ref{thm:sybil:lookahead}}
Let us first formalize the concept of sybils. Given a player $i$ with stake $\sigma_i$ and type $\tau_i$, we say that $i$ can split into sybils $(\xi_1, \ldots, \xi_{\ell})$ where each sybil $\xi_j$ has a stake $\sigma_{\xi_j}$ and a type $\tau_{\xi_j}$ such that $\sum_{j = 1}^\ell \sigma_{\xi_j} \leq \sigma_i$ and $\sum_{j = 1}^\ell \tau_{\xi_j} \leq \tau_i$. We say that a policy is \emph{sybil-proof} if for each stake profile $\tm{\sigma}$, for each player $i$, and every possible partition in sybils $(\xi_1, \ldots, \xi_\ell)$, the expected utility that $i$ receives from this policy when the current stake profile is $\tm{\sigma}$ and current type profile is $\tau$, is at least as large as the sum of the expected utility that the sybils receive from the policy if the current stake profile is $(\tm{\sigma}_{-i}, \sigma_{\xi_1}, \ldots, \sigma_{\xi_\ell})$ and the current type profile is $(\tau_{-i}, \tau_{\xi_1}, \ldots, \tau_{\xi_\ell})$.

It is not hard to see that $\mu^\ell$ is not sybil-proof. Suppose indeed that $\tm{\sigma}$ is harmful for player 1. We say that this player has \emph{recovery sybils} if there are $(\xi_1, \ldots, \xi_{\ell})$ such that $(\tm{\sigma}_{-i}, \sigma_{\xi_1}, \ldots, \sigma_{\xi_\ell})$ is not harmful for the sybil with the largest type. Note that recovery sybils always exist (simply evenly split in $\tau_i$ sybils with {minimum} type and equal stake). The \emph{preferred} recovery sybils $(\xi^*_1, \ldots, \xi^*_{\ell})$ are the ones that maximize the maximum type among sybils. Then if the maximum type among $(\xi^*_1, \ldots, \xi^*_{\ell})$ is larger than (or it is the same as, but the associated stake is larger than) the one of the player 2, then we conclude that player 1 will prefer to create sybils. Indeed, without sybils either this player does not participate, and so she does not receive any reward, or, if she participates, then the loss in the system value will  outweigh the (expected) reward. On the other hand, if sybils participate in place of player 1, then they will receive an expected reward that is at least as large as 
the expected reward that player $1$ would have received upon participation, but without affecting the system value.

\begin{theorem}\label{thm:sybil:lookahead}
 Policy $\mu^\ell$ is sybil-proof if for every player $i$ and every $\sigma$ that is harmful for $i$ it holds that the corresponding preferred recovery sybils are such that the largest type is less than (or it is the same as, but the associated stake is smaller than) the one of player $i+1$.
\end{theorem}
\begin{proof}
 If player $i$ splits into the preferred recovery sybils, then the expected gain of these players must be 0, since in this case no sybil has the largest type. If player $i$ splits in other recovery sybils the situation cannot be better for the same reason. If player $i$ splits in non-recovery sybils, then participation of sybils would be harmful for the sybil with largest type, and thus the expected gain of this player will be negative. Hence, the player cannot achieve gain larger than zero, that is exactly the one that is achieved by player $i$ if she does not participate.   
\end{proof}

%%%%%%%%%%%%%%%%%%%%%%%%%%%%%%%%%%%%%%%%%%%%%%%%%%%%%%%%%%%%%%%%%%%%%%%%

\end{document}